\documentclass[a4paper,onecolumn,11pt,aps,nofootinbib,superscriptaddress,tightenlines]{revtex4}

\usepackage{amsmath}
\usepackage{amssymb}  
\usepackage{amsthm}
\usepackage{amsfonts}
\usepackage{graphicx}
\usepackage{bbm}
\usepackage{url}
\usepackage{ upgreek }
\usepackage{dsfont}
\usepackage{color}

\newcommand{\tr}{\textnormal{tr}}

\newcommand{\ket}[1]{| #1 \rangle}

\newcommand{\bra}[1]{\langle #1 |}

\newcommand{\braket}[2]{\langle #1 | #2 \rangle}

\newcommand{\proj}[2]{| #1 \rangle\!\langle #2 |}
\newcommand{\proji}[3]{| #1 \rangle\!\langle #2 |_{#3}}

\newcommand{\id}{\ensuremath{\mathds{1}}}

\def\beq{\begin{equation}}
\def\eeq{\end{equation}}
\def\bq{\begin{quote}}
\def\eq{\end{quote}}
\def\ben{\begin{enumerate}}
\def\een{\end{enumerate}}
\def\bit{\begin{itemize}}
\def\eit{\end{itemize}}

\def\ra{\rightarrow}

\def\lb{\left(}
\def\rb{\right)}
\def\lset{\lbrace}
\def\rset{\rbrace}

\def\l|{\left|}
\def\r|{\right|}
\def\lbr{\left[}
\def\rbr{\right]}
\def\ident{\textnormal{id}}
\def\one{\id}

\newcommand\C{\mathbbm{C}}

\newcommand\R{\mathbbm{R}}
\newcommand\N{\mathbbm{N}}
\newcommand\M{\mathcal{M}}
\newcommand\D{\mathcal{D}}

\newcommand{\Tm}{\mathcal{T}}

\newcommand{\ketbra}[1]{|#1\rangle\langle#1|}

\theoremstyle{plain}
\newtheorem{thm}{Theorem}[section]
\newtheorem{lem}{Lemma}[section]
\newtheorem{cor}{Corollary}[section]

\newtheorem{defn}{Definition}[section]
\theoremstyle{definition}

\begin{document}
\title{\vspace{-1.0cm}{\textbf{Relative Entropy Bounds on Quantum, Private and Repeater Capacities}}}

\author{Matthias Christandl}
\email{christandl@math.ku.dk}
\affiliation{QMATH, Department of Mathematical Sciences, University of Copenhagen, Universitetsparken 5, 2100 Copenhagen $\o$, Denmark}

\author{Alexander M\"uller-Hermes}
\email{muellerh@posteo.net, muellerh@math.ku.dk}
\affiliation{QMATH, Department of Mathematical Sciences, University of Copenhagen, Universitetsparken 5, 2100 Copenhagen $\o$, Denmark}
\begin{abstract}
We find a strong-converse bound on the private capacity of a quantum channel assisted by unlimited two-way classical communication. The bound is based on the max-relative entropy of entanglement and its proof uses a new inequality for the sandwiched R\'{e}nyi divergences based on complex interpolation techniques. We provide explicit examples of quantum channels where our bound improves upon both the transposition bound (on the quantum capacity assisted by classical communication) and the bound based on the squashed entanglement. As an application we study a repeater version of the private capacity assisted by classical communication and provide an example of a quantum channel with high private capacity but negligible private repeater capacity. 
\end{abstract}
\maketitle
\date{\today}

\tableofcontents
\section{Introduction}

The goal of Shannon theory~\cite{Shannon1948} is to quantify the amount of information that can be reliably transmitted using many copies of a communication channel. To protect the information from errors induced by the channel, particular coding schemes may be applied. For a given class of coding schemes a capacity can be defined quantifying the optimal rate of reliable information transmission achievable using schemes from the class. In quantum Shannon theory there are many different capacities describing relevant coding scenarios where certain types of classical or quantum assistance are allowed. Here we are interested in capacities where arbitrary classical communication between the two communicating parties is allowed to assist the transmission of quantum or private information. 

For a quantum channel $T:\M_{d_A}\ra\M_{d_B}$ we denote by $\mathcal{Q}_{\leftrightarrow}(T)$ ($\mathcal{P}_\leftrightarrow(T)$) its quantum (private) capacity assisted by two-way classical communication. While it is true that $\mathcal{P}_{\leftrightarrow}(T)$ is an upper bound on $\mathcal{Q}_{\leftrightarrow}(T)$ it is important to have simpler upper bounds in terms of single-letter quantities only depending on the quantum channel $T$. Not many such bounds on  $\mathcal{Q}_{\leftrightarrow}$ and $\mathcal{P}_{\leftrightarrow}$ are known: In \cite{takeoka2014squashed} the squashed entanglement of a quantum channel has been defined and shown to be an upper bound on $\mathcal{P}_{\leftrightarrow}$ (and therefore also on $\mathcal{Q}_{\leftrightarrow}$). The transposition bound (see \cite{holevo2001evaluating}) has been shown to be a strong-converse bound on $\mathcal{Q}_{\leftrightarrow}$ in \cite{muller2016positivity}. Finally, in \cite{berta2013entanglement} the entanglement cost of a quantum channel has been defined and shown to be a strong-converse bound on $\mathcal{Q}_{\leftrightarrow}$.  
  
For particular classes of channels other upper bounds are known. Recently the class of teleportation covariant channels has received much attention in this context~\cite{pirandola2015ultimate,pirandola2016optimal,
takeoka2016unconstrained,wilde2016converse,pant2016rate,laurenza2016general}. Special cases of such channels have been considered in \cite{bennett1996mixed}, and recently more relevant examples have been identified. In particular this family contains the Gaussian channels in infinite dimensions as an important subclass~\cite{pirandola2015ultimate}. We will be interested mostly in the finite-dimensional case. For a finite-dimensional teleportation covariant channel $T$ the capacity $\mathcal{P}_{\leftrightarrow}(T)$ is equal to the distillable key of the Choi-Jamiolkowski state $C_T$~\cite{choi1975completely} corresponding to the channel (see teleportation stretching~\cite{pirandola2015ultimate} for a generalization of these arguments to the case of infinite-dimensional quantum channels). Using that the relative entropy of entanglement $E_R$ is an upper bound on the distillable key~\cite{horodecki2005secure} any finite-dimensional teleportation covariant channel fulfills the bound (see~\cite{pirandola2015ultimate})  
\begin{equation}
\mathcal{P}_{\leftrightarrow}(T) = \mathcal{K}_{\leftrightarrow}(C_T)\leq E_R(C_T)
\label{equ:Pirandola}
\end{equation}
and this is also a strong-converse bound (see \cite{wilde2016converse}). It is still an open problem whether a similar bound based on the relative entropy of entanglement (possibly involving an optimization over the input state of the partial channel) holds for arbitrary quantum channels $T$. 

In this article we establish an upper bound on $\mathcal{P}_{\leftrightarrow}$ for arbitrary quantum channels in terms of the max-relative entropy of entanglement. Given a quantum channel $T:\M_{d_A}\ra\M_{d_B}$ its max-relative entropy of entanglement is defined as
\begin{equation}
E_{\max}(T) = \sup\lset E^{A':B}_{\max}\lb T^{A\ra B}\lb \rho_{A'A}\rb\rb ~:~ \rho_{A'A}\in\D\lb\C^{d_{A'}}\otimes \C^{d_A}\rb, d_{A'}\in\N\rset.
\end{equation}
Here $E^{A':B}_{\max}$ denotes the max-relative entropy of entanglement of states~\cite{datta2009min,datta2009max}. Our paper is structured as follows:
\begin{itemize}
\item In Section \ref{sec:dataProc} we use complex interpolation techniques to prove a new inequality (the ``data-processed triangle inequality'') for the sandwiched $\alpha$-R\'{e}nyi divergence (see Section \ref{sec:SandWich} for a definition).  
\item Using the data-processed triangle inequality we show in Section \ref{sec:StrongConverse} that for any quantum channel $T:\M_{d_A}\ra \M_{d_B}$ the quantity $E_{\max}(T)$ is a strong-converse bound on $\mathcal{P}_{\leftrightarrow}(T)$. 
\item In Section \ref{sec:Lock} we show that $E_{\max}(T)$ is non-lockable (see Corollary \ref{cor:RedChanUp} for the precise statement). We use this feature of our bound in Section \ref{sec:Flower} to give examples of channels, where our bound improves upon the previously known bounds (transposition bound, squashed entanglement bound and entanglement cost).
\item In Section \ref{sec:SimplUpp} we give a weaker upper bound on $\mathcal{P}_{\leftrightarrow}(T)$ for any quantum channel $T:\M_{d_A}\ra \M_{d_B}$ that is slightly easier to evaluate than our original $E_{\max}$ bound. As an application we then study a repeater version of the private capacity in Section \ref{sec:non-rep}, where the communicating parties can use an intermediate repeater station to perform private communication. We show that there are quantum channels $T$ which have a high private capacity, but where the repeated private capacity can be arbitrarily close to zero. This is the channel version of a result demonstrated in \cite{bauml2015limitations} where states connecting the three parties are given.   
\item In the Appendix we give an example of a quantum channel that cannot be implemented via an LOCC-protocol from any state preparable by a single use of the channel (see Definition \ref{defn:ImpleFrIm}). This property is needed to obtain a bound similar to \eqref{equ:Pirandola} based on the relative entropy of entanglement using the arguments of \cite{pirandola2015ultimate}.      
\end{itemize}

\section{Preliminaries}

In the following we denote the complex $d\times d$-matrices by $\M_d$ and the cone of positive matrices by $\M^+_d$. The $d\times d$ identity matrix is denoted by $\one_d$. The set of $d\times d$ quantum states (i.e.\ positive $d\times d$ matrices with trace $1$) is called $\D_d = \D\lb\C^d\rb$. Pure states will be denoted as projectors using the notation $\ketbra{\psi}\in\D_d$ for $\ket{\psi}\in\C^d$ with $\braket{\psi}{\psi} = 1$. On multipartite systems we will often use indices $A,B,\ldots$ to indicate the different tensor factors. For example we would write $\rho_{ABC}\in\D\lb \C^{d_A}\otimes \C^{d_B}\otimes \C^{d_C}\rb$ for a tripartite state. We use the common notation of omitting indices to denote partial traces (i.e.\ the state $\rho_{A}$ would be the marginal of $\rho_{ABC}$ on the $A$ system). For general linear maps $T:\M_{d_A}\ra\M_{d_B}$ we write $T^{A\ra B}\lb\rho_{AA'}\rb\in\D\lb\C^{d_B}\otimes\C^{d_{A'}}\rb$ to denote its partial application to the $A$ system of the state $\rho_{AA'}$. In this sense the Choi matrix~\cite{choi1975completely} of a linear map $T:\M_{d_A}\ra\M_{d_B}$ is denoted by 
\begin{equation}
C_T = T^{A\ra B}\lb\omega_{A'A}\rb, 
\label{equ:Choi}
\end{equation}
where $\omega_{A'A}\in\D\lb\C^{d_{A'}}\otimes \C^{d_A}\rb$ for $d_A=d_{A'}$ denotes the maximally entangled state in the computational basis (i.e.\ $\omega_{A'A} = \ketbra{\Omega_{A'A}}$ for $\ket{\Omega_{A'A}} = \frac{1}{\sqrt{d_A}}\sum^{d_A}_{i=1}\ket{i_{A'}i_{A}}$ ). We will also use the notation $\omega_d\in\D\lb\C^{d}\otimes \C^d\rb$ to denote this state in the cases where the concrete systems are not important. Most linear maps we will use are quantum channels (i.e.\ trace-preserving and completely positive~\cite{choi1975completely}). A well-known example of a positive, but not completely-positive, map is the transposition $\vartheta_d:\M_d\ra\M_d$ given by $\vartheta_d(X)=X^T$ in the computational basis. We will also use the notation $\vartheta_A$ to denote the partial transposition on a particular system (named $A$ in this case).

\subsection{Sandwiched $\alpha$-R\'{e}nyi divergences}
\label{sec:SandWich}

For quantum states $\rho,\sigma\in\D\lb\C^d\rb$ and a parameter $\alpha\in(1,\infty)$, the \emph{sandwiched $\alpha$-R\'{e}nyi divergence} \cite{muller2013quantum,wilde2014strong} is defined as
\begin{equation}
 D_\alpha\lb \rho\| \sigma\rb=
\begin{cases} 
\frac{1}{\alpha-1}\log\lb\tr\left[\lb \sigma^{\frac{1-\alpha}{2\alpha}}\rho \sigma^{\frac{1-\alpha}{2\alpha}}\rb^\alpha\right]\rb, & \mbox{if }  
\text{supp}[\rho]\subseteq\text{supp}[\sigma]
\\ +\infty, & \mbox{otherwise}.
\end{cases}
\label{equ:Renyi}
\end{equation}
In \cite{muller2013quantum} it has been shown that the limiting cases $\alpha=1$ and $\alpha=\infty$ of $D_\alpha$ coincide with quantities studied before: In the limit $\alpha\ra 1$ we have 
\[
D_\alpha\lb \rho\| \sigma\rb\ra D\lb \rho\|\sigma\rb = \text{tr}\lbr\rho\lb \log(\rho) - \log(\sigma)\rb\rbr
\]
which is the usual relative entropy~\cite{umegaki1962}. We will sometimes write $D_1$ to denote the relative entropy. Taking the limit $\alpha\ra\infty$ gives $D_\alpha\lb \rho\| \sigma\rb\ra D_{\max}\lb \rho\| \sigma\rb$ which is the max-relative entropy~\cite{datta2009min}. For quantum states $\rho,\sigma\in\D\lb\C^d\rb$ this quantity can be defined in two equivalent ways as 
\begin{equation}
D_{\max}\lb \rho\| \sigma\rb = \inf\lset\lambda\in\R^+ : \rho\leq 2^{\lambda}\sigma \rset = \begin{cases} 
\log\lb \| \sigma^{-\frac{1}{2}}\rho \sigma^{-\frac{1}{2}}\|_\infty\rb & \mbox{if }  
\text{supp}[\rho]\subseteq\text{supp}[\sigma]
\\ +\infty, & \mbox{otherwise}
\end{cases}
\label{equ:defDmax}
\end{equation}
using the convention $\inf \emptyset = +\infty$.  

In~\cite{beigi} it has been noted that the sandwiched $\alpha$-R\'{e}nyi divergence $D_\alpha$ (see \eqref{equ:Renyi}) for $\alpha > 1$ can be written in terms of a non-commutative $L_{\alpha,\sigma}$-norm $\|\cdot\|_{\alpha,\sigma}$ defined as 
\[
\| X\|_{\alpha,\sigma} = \tr\lbr\left|\sigma^{\frac{1}{2\alpha}}X\sigma^{\frac{1}{2\alpha}}\right|^\alpha\rbr^{\frac{1}{\alpha}}
\]
for any $X\in\M_d$ and $\sigma\in\M^+_d$. With the function $\Gamma_\sigma:\M_d\ra\M_d$ given by $\Gamma_\sigma(X) = \sigma^{1/2}X\sigma^{1/2}$ we can write 
\begin{equation}
D_\alpha\lb\rho\|\sigma\rb = \frac{1}{\alpha-1}\log\lb\| \Gamma_\sigma^{-1}\lb\rho\rb\|_{\alpha,\sigma}^\alpha\rb
\label{equ:SandwichNormVersion}
\end{equation}
for any quantum states $\rho,\sigma\in\D\lb\C^d\rb$ with $\text{supp}\lb\rho\rb\subseteq\text{supp}\lb\sigma\rb$ using the Moore-Penrose pseudo-inverse~\cite{penrose1955generalized} in the case where $\sigma$ is not full-rank. 

For a linear map $L:\M_{d_1}\ra\M_{d_2}$ we will use norms of the form
\begin{equation}
\| L\|_{(p,\sigma)\ra (q,\sigma')} = \sup_{X\in\M_{d_1}}\frac{\|L(X)\|_{q,\sigma'}}{\| X\|_{p,\sigma}},
\label{equ:ptoqNorms}
\end{equation}
which are the operator norms of the operator $L$ as a mapping from the space $(\M_{d_1},\|\cdot\|_{p,\sigma})$ to $(\M_{d_2},\|\cdot\|_{q,\sigma'})$. For $\sigma = \one_{d_1}$ and $\sigma' = \one_{d_2}$ the above definition gives the usual $p\ra q$-norms and we will use the common notation $\|\cdot\|_{p\ra q}$ in this case. The main technical tool we will use, is the following non-commutative Riesz-Thorin-type theorem. It should be noted that similar interpolation theorems have a long history (see~\cite{bergh2012interpolation}).

\begin{thm}[Riesz-Thorin Theorem for $L_{p,\sigma}$ spaces~\cite{beigi}]
Let $L:\M_{d_1}\ra\M_{d_2}$ be a linear map. For $1\leq p_0\leq p_1\leq \infty$ and $1\leq q_0\leq q_1\leq \infty$ and $\theta\in (0,1)$ we define $p_\theta$ via
\[
\frac{1}{p_\theta} = \frac{\theta}{p_0} + \frac{1-\theta}{p_1}.
\]
and $q_\theta$ analogous. Then for positive definite matrices $\sigma\in\M^+_{d_1}$ and $\sigma'\in\M^+_{d_2}$ we have
\[
\| L\|_{(p_\theta,\sigma)\ra (q_\theta,\sigma')} \leq \| L\|^\theta_{(p_0,\sigma)\ra (q_0,\sigma')} \| L\|^{1-\theta}_{(p_1,\sigma)\ra (q_1,\sigma')} .
\]
\label{thm:RieszT}
\end{thm}

A consequence of the previous theorem is the monotonicity of the sandwiched $\alpha$-R\'{e}nyi divergences under quantum channels for $\alpha >1$ (see \cite{beigi}), i.e.\ the inequality 
\begin{equation}
D_{\alpha}\lb T(\rho)\| T(\sigma)\rb\leq D_\alpha\lb\rho\|\sigma\rb
\label{equ:DataProc}
\end{equation} 
for any quantum channel $T:\M_{d_1}\ra\M_{d_2}$ and quantum states $\rho,\sigma\in\D_{d_A}$. Inequality~\eqref{equ:DataProc} also holds for trace-preserving positive maps $T$ as shown in \cite{muller2015monotonicity} and for quantum channels when $\alpha\geq\frac{1}{2}$  \cite{frank2013monotonicity, muller2013quantum}.

\subsection{$\alpha$-Relative entropies of entanglement and related measures}
 
For any $\alpha\geq 1$ we can introduce an $\alpha$-relative entropy of entanglement generalizing the usual relative entropy of entanglement (also introduced recently in \cite{wilde2016converse}). 

\begin{defn}[$\alpha$-Relative Entropy of Entanglement]
For a bipartite quantum state $\rho_{AB}\in\D(\C^{d_A}\otimes \C^{d_B})$ we define the $\alpha$-relative entropy of entanglement as 
\[
E^{A:B}_{\alpha}(\rho_{AB}) = \min\lset D_{\alpha}(\rho_{AB}\|\sigma_{AB}) ~:~ \sigma_{AB}\in\text{Sep}_{A:B}\lb \C^{d_A}\otimes \C^{d_B}\rb\rset
\]
where $\text{Sep}_{A:B}\lb \C^{d_A}\otimes \C^{d_B}\rb$ denotes the set of separable states w.r.t. the bipartition $A:B$.
\label{defn:alphaRelEntrEnt} 
\end{defn}
Using the convergence of $D_\alpha$ it is clear that $E^{A:B}_{\alpha}\ra E^{A:B}_R$ as $\alpha\ra 1$ for the relative entropy of entanglement denoted by $E_R$. Similarly we can take the limit $\alpha\ra\infty$ and obtain the max-relative entropy of entanglement\footnote{also known as log-robustness \cite{brandao2010reversible}.} 
\[
E^{A:B}_{\max}(\rho_{AB}) = \min\lset D_{\max}(\rho_{AB}\|\sigma_{AB}) ~:~ \sigma_{AB}\in\text{Sep}_{A:B}\lb \C^{d_A}\otimes \C^{d_B}\rb\rset, 
\]
which has been studied in \cite{datta2009min,datta2009max,brandao2010generalization}. For any $\alpha\geq 1$  the $\alpha$-relative entropy of entanglement can be used to quantify the transmission of entanglement over a quantum channel. We will focus on the case $\alpha = \infty$ and the following quantity (also recently introduced in \cite{wilde2016converse}): 

\begin{defn}[max-relative entropy of entanglement of a quantum channel]\hfill \\
For a quantum channel $T:\M_{d_A}\ra\M_{d_B}$ we define the max-relative entropy of entanglement of T as 
\begin{equation}
E_{\max}(T) = \sup\lset E^{A':B}_{\max}\lb T^{A\ra B}\lb \rho_{A'A}\rb\rb ~:~ \rho_{A'A}\in\D\lb\C^{d_{A'}}\otimes \C^{d_A}\rb, d_{A'}\in\N\rset.
\label{equ:maxRelEntrQChan}
\end{equation}
\end{defn}
Using quasi-convexity of $D_{\max}$ (see \cite[Lemma 9]{datta2009min}) and the Schmidt-decomposition of pure quantum states it is not hard to show, that the dimension $d_{A'}$ appearing in the supremum can be chosen as the input dimension of the quantum channel. More specifically, for any quantum channel $T:\M_{d_A}\ra\M_{d_B}$ we get the following equivalent expression
\[
E_{\max}\lb T\rb = \max\lset E^{A':B}_{\max}\lb T^{A\ra B}\lb \ketbra{\psi_{A'A}}\rb\rb ~:~ \ketbra{\psi_{A'A}}\in\D\lb C^{d_{A'}}\otimes \C^{d_A}\rb\text{ for } d_{A'} = d_A\rset. 
\] 
In particular this shows that the max-relative entropy of a quantum channel is well-defined and we will use a $\max$ instead of the $\sup$ in \eqref{equ:maxRelEntrQChan} to indicate that the optimum is attained.

\subsection{Quantum capacities assisted by classical communication}
   
A quantum channel on bipartite systems $L:\M_{d_A}\otimes\M_{d_{B}}\to\M_{d_{A'}}\otimes\M_{d_{B'}}$ is called implementable via local operations and classical communications (w.r.t.\ bipartitions $A:B$ and $A':B'$ of the input and output systems, respectively) if it can be written as a composition of any number of channels $L_{A_q:B_q\to A'_qA'_c:B'_qB'_c}$ of the following form ($X_{A_qB_q}\in\M_{d_{A_q}}\otimes\M_{d_{B_q}}$):
\begin{align}\label{LOCCkraus}
L_{A_q:B_q\to A'_qA'_c:B'_qB'_c}(X_{A_qB_q})=\sum_{i,j}(K^A_i\otimes K^B_j)X_{A_qB_q}(K^A_i\otimes K^B_j)^\dagger\otimes\ket{j}\bra{j}_{A'_c}\otimes\ket{i}\bra{i}_{B'_c}.
\end{align}
Here $K^A_i:\C^{|A_q|}\to\C^{|A'_q|}$ and $K^B_j:\C^{|B_q|}\to\C^{|B'_q|}$ $(i\in I, j\in J)$ are Kraus operators of quantum channels mapping system $A_q$ to $A'_q$ and system $B_q$ to $B'_q$ respectively (i.e.\ $\sum_i(K^A_i)^\dagger K^A_i=\id_{A_q}$ and $\sum_j(K^B_j)^\dagger K^B_j=\id_{B_q}$), and $\ket{j}_{A'_c}$ and $\ket{i}_{B'_c}$ are orthonormal bases belonging to (effectively classical) systems $A_c$ and $B_c$ of dimension $|J|$ and $|I|$ (see \cite{chitambar2014everything} for more details). In the following we will call a quantum channel implementable via local operations and classical communications simply an LOCC-operation. 

\begin{figure}
\includegraphics[scale=0.75]{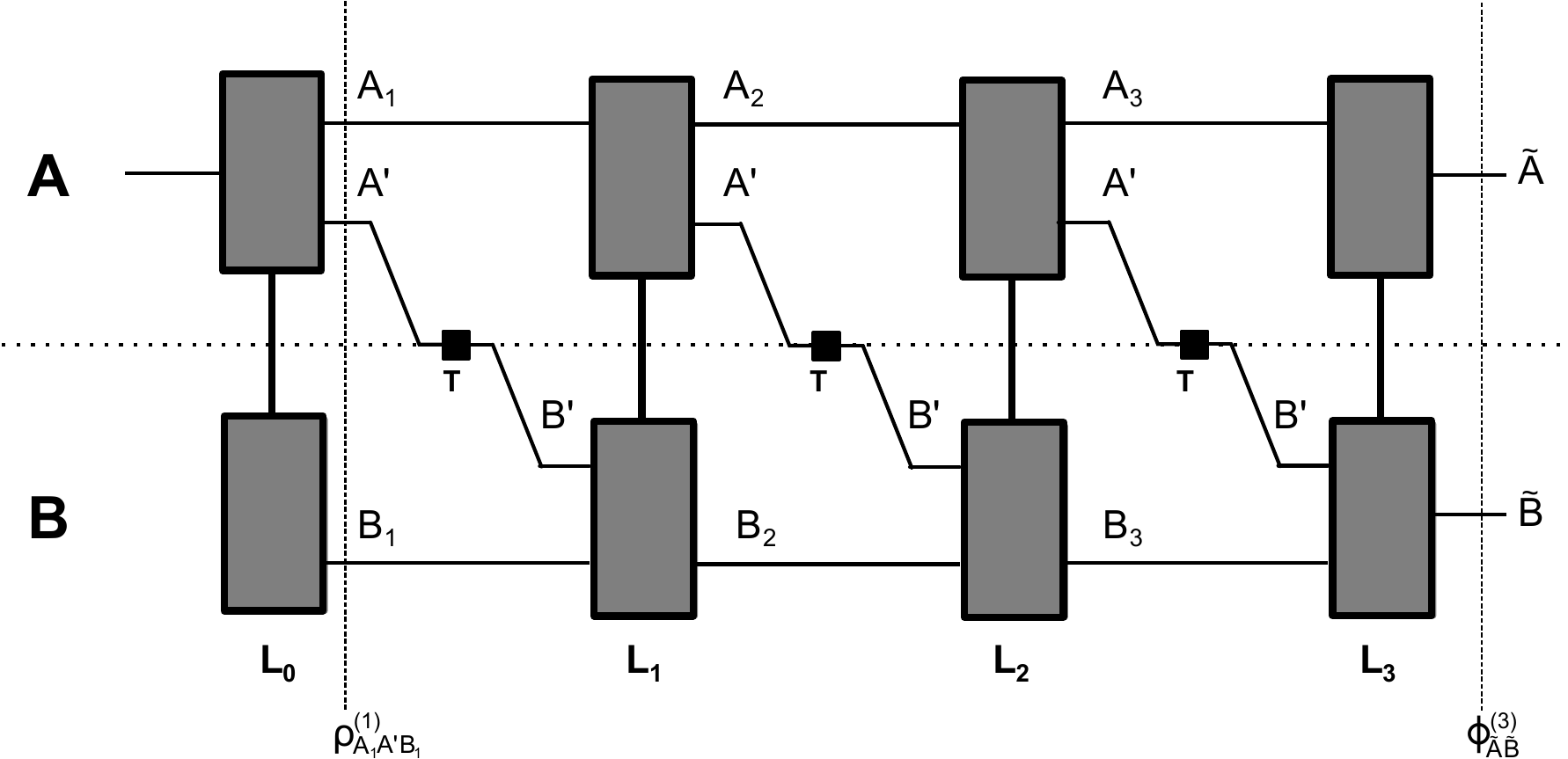}
\caption{Coding scheme assisted by classical communication (cf. Definition \ref{defn:CodingScheme}) in the case of $m=3$ uses of the channel $T:\M_{d_{A'}}\ra\M_{d_{B'}}$. Here $L_0$ denotes an LOCC-operation used to create the separable initial state $\rho^{(1)}_{A_1 A' B_1}$.}
\label{fig:CodingScheme}
\end{figure}

We can now define coding schemes assisted by classical communication:

\begin{defn}[Coding schemes assisted by classical communication]\hfill \\
Let $T:\M_{d_{A'}}\to\M_{d_{B'}}$ be a quantum channel. A coding scheme assisted by classical communication with $m$ uses of the channel $T$ is given by a separable initial state 
\[
\rho^{(1)}_{A_1 A' B_1}\in\text{Sep}_{A_1 A':B_1}\lb\C^{d_{A_1}d_{A'}}\otimes \C^{d_{B_1}}\rb
\]
and a set of LOCC-operations $\lset L_i\rset^{m}_{i=1}$ (see also Figure \ref{fig:CodingScheme}). Here
\[
L_i:\M_{d_{A_i}}\otimes \M_{d_{B'}}\otimes \M_{d_{B_i}}\ra \M_{d_{A_{i+1}}}\otimes \M_{d_{A'}}\otimes \M_{d_{B_{i+1}}}
\]  
for each $i\in\lset 1,\ldots ,m-1\rset$ and 
\[
L_{m}:\M_{d_{A_m}}\otimes \M_{d_{B'}}\otimes \M_{d_{B_m}}\ra \M_{d_{\tilde{A}}}\otimes \M_{d_{\tilde{B}}}
\]
are LOCC w.r.t. the bipartition into $A$ and $B$ systems for arbitrary dimensions $d_{A_i}, d_{B_i}, d_{\tilde{A}} , d_{\tilde{B}}$. The output state of the coding scheme will be denoted by 
\[
\phi^{(m)}_{\tilde{A}\tilde{B}} = L_m\circ\prod^{m-1}_{i=1} \lb T^{A'\ra B'}\circ L_i\rb\circ T^{A'\ra B'}\lb\rho^{(1)}_{A_1 A' B_1}\rb.
\]

\label{defn:CodingScheme}
\end{defn}

We will first state the definition of the quantum capacity assisted by two-way classical communication. In the presence of unlimited classical communication we can use quantum teleportation~\cite{bennett1993teleporting} to turn any entanglement generation protocol into a quantum communication protocol. Therefore, we can define the quantum capacity assisted by two-way communication in terms of entanglement generation.

\begin{defn}[Entanglement generation assisted by classical communication]\hfill \\
Given a quantum channel $T:\M_{d_{A'}}\to\M_{d_{B'}}$ consider a coding scheme assisted by classical communication with $m$ channel uses (as in Definition \ref{defn:CodingScheme}) given by LOCC-operations $\lset L_i\rset^{m+1}_{i=1}$, initial state $\rho^{(1)}_{A_1 A' B_1}$ and output state $\phi^{(m)}_{\tilde{A}\tilde{B}}\in\D\lb\C^{d_{\tilde{A}}}\otimes \C^{d_{\tilde{B}}}\rb$. Such a coding scheme is called an $(n,m,\epsilon)$-coding scheme for entanglement generation assisted by classical communication iff the output dimensions fulfill $d_{\tilde{A}} = d_{\tilde{B}} = 2^n$ and the output state satisfies
\[
\epsilon = \frac{1}{2}\|\phi^{(m)}_{\tilde{A}\tilde{B}} - \omega_{2^n}\|_1.
\] 

\label{defn:CodingSchemeEG}
\end{defn}

\begin{defn}[Quantum capacity assisted by classical communication]

We call $R\in\R^+$ an achievable rate for quantum communication over the channel $T$ assisted by classical communication iff for each $\nu\in\N$ there exists a $(n_\nu,m_\nu,\varepsilon_\nu)$-coding scheme for entanglement generation assisted by classical communication (as in Definition \ref{defn:CodingSchemeEG}) with $m_\nu\ra\infty$ as $\nu\ra\infty$ such that $R=\lim_{\nu\to\infty}\frac{n_\nu}{m_\nu}$ and $\lim_{\nu\to\infty}\varepsilon_\nu=0$. The quantum capacity of $T$ assisted by classical two-way communication ${\mathcal{Q}}_\leftrightarrow(\Tm)$ is defined to be the supremum of all such achievable rates.

\end{defn}

In a similar way we can define the private capacity assisted by classical two-way communication. It has been shown in \cite{horodecki2005secure} that the tasks of private communication using a quantum channel and public communication is equivalent to the task of distilling private states using a coding scheme assisted by classical communication (see also \cite{wilde2016converse}). We will begin by defining these states:

\begin{defn}[Private states~\cite{horodecki2005secure}]
A quantum state 
\[
\gamma_{A_kB_kA_sB_s}\in\D\lb\C^{d_{A_k}}\otimes \C^{d_{B_k}}\otimes \C^{d_{A_s}}\otimes \C^{d_{B_s}}\rb
\] 
with $d_{A_k} = d_{B_k} = K$ and $d_{A_s}=d_{B_s}$ is called a private state with $K$-dimensional key part iff it is of the form
\[
\gamma_{A_kB_kA_sB_s} = U^{\text{tw}}_{A_kB_kA_sB_s}\lb\omega_{A_kB_k}\otimes \sigma_{A_sB_s}\rb (U^{\text{tw}}_{A_kB_kA_sB_s})^\dagger 
\]
for some quantum state $\sigma_{A_sB_s}\in\D\lb\C^{d_{A_s}}\otimes \C^{d_{B_s}}\rb$ where we applied a twisting unitary of the form
\[
U^{\text{tw}}_{A_kB_kA_sB_s} = \sum^{d_{A_k}}_{i=1}\sum^{d_{B_k}}_{j=1} \ketbra{i}_{A_k}\otimes \ketbra{j}_{B_k}\otimes U^{ij}_{A_sB_s}
\]
with $U^{ij}_{A_sB_s}\in\mathcal{U}_{d_{A_s}d_{B_s}}$ unitary for any $i,j$. The systems $A_k,B_k$ are called the key part and $A_s,B_s$ the shield part of the private state.  
\label{defn:PrivState}
\end{defn}

It can be shown (see~\cite{horodecki2005secure}) that any private state with $K$-dimensional key part held by two parties $A$ and $B$ can be used to generate at least $\log_2(K)$ secret bits shared between the two parties (protected from any eavesdropper who might possess the purification of the state). Note that in the above definition there might be more than $\log_2(K)$ secret bits obtainable (i.e.\ the private state is not necessarily irreducible~\cite{horodecki2005secure}).   

Now we can define the private capacity assisted by classical communication as a private state generation capacity.

\begin{defn}[Coding scheme for private state generation assisted by classical communication]
Given a quantum channel $T:\M_{d_{A'}}\to\M_{d_{B'}}$ consider a coding scheme assisted by classical communication with $m$ channel uses (as in Definition \ref{defn:CodingScheme}) given by LOCC-operations $\lset L_i\rset^{m+1}_{i=1}$, initial state $\rho^{(1)}_{A_1 A' B_1}$ and output state $\phi^{(m)}_{\tilde{A}\tilde{B}}\in\D\lb\C^{d_{\tilde{A}}}\otimes \C^{d_{\tilde{B}}}\rb$. Such a coding scheme is called a $(k,m,\epsilon)$-coding scheme for private state generation assisted by classical communication iff the output dimensions factorize into $d_{\tilde{A}} = d_{A_k}d_{A_s}$ and $d_{\tilde{B}} = d_{B_k}d_{B_s}$ for $d_{A_k} = d_{B_k} = 2^k$ and $d_{A_s} = d_{B_s}$, and the output state satisfies
\[
\epsilon = \frac{1}{2}\|\phi^{(m)}_{\tilde{A}\tilde{B}} - \gamma_{A_kB_kA_sB_s}\|_1.
\]
for a private state $\gamma_{A_kB_kA_sB_s}\in\D\lb\C^{d_{\tilde{A}}}\otimes \C^{d_{\tilde{B}}}\rb$ with $2^k$-dimensional key part. 

\label{defn:PrivCodingScheme}
\end{defn}

\begin{defn}[Private capacity assisted by classical communication]
We call $R\in\R^+$ an achievable rate for private communication over the channel $T$ assisted by classical communication iff for each $\nu\in\N$ there exists a $(k_\nu,m_\nu,\varepsilon_\nu)$-coding scheme for private state generation assisted by classical communication (as in Definition \ref{defn:PrivCodingScheme}) with $m_\nu\ra\infty$ as $\nu\ra\infty$ such that $R=\lim_{\nu\to\infty}\frac{k_\nu}{m_\nu}$ and $\lim_{\nu\to\infty}\varepsilon_\nu=0$. The private quantum capacity of $T$ assisted by classical two-way communication ${\mathcal{P}}_\leftrightarrow(\Tm)$ is defined to be the supremum of all such achievable rates.
\label{defn:PrivCap}
\end{defn}

In the remaining part of this section we will discuss some general upper and strong-converse bounds on $\mathcal{Q}_\leftrightarrow$ and $\mathcal{P}_\leftrightarrow$. Recall that an upper bound $B\geq 0$ on either $\mathcal{Q}_\leftrightarrow$ or $\mathcal{P}_\leftrightarrow$ is called a \emph{strong converse bound} iff for any sequence of $(n_\nu,m_\nu,\varepsilon_\nu)$-coding schemes (for $\nu\in\N$) leading to a rate $R = \lim_{\nu\to\infty}\frac{n_\nu}{m_\nu} > B$ the error fulfills $\epsilon_\nu\ra 1$ as $\nu\ra\infty$.  We will start with the transposition bound (originally introduced in~\cite{holevo2001evaluating}), based on the matrix transposition $\vartheta_d :\M_{d}\ra\M_{d}$, i.e.\ $\vartheta_d(X) = X^T$ in any fixed basis.

\begin{thm}[Transposition bound~\cite{muller2016positivity}]
For any quantum channel $T:\M_{d_1}\ra\M_{d_2}$ we have
\[
\mathcal{Q}_\leftrightarrow(T)\leq \log_2\lb\|\vartheta_{d_2}\circ T\|_\diamond\rb
\]
and the above bound is a strong-converse bound.
\label{thm:transpBound}
\end{thm} 

Another bound is based on the squashed entanglement introduced in~\cite{christandl2004squashed,tucci2002entanglement}. Recall the definition of the quantum conditional mutual information of a tripartite quantum state $\rho_{ABE}\in\D\lb\C^{d_A}\otimes\C^{d_B}\otimes \C^{d_E}\rb$ given by 
\[
I(A;B|E)_{\rho_{ABE}} = S(\rho_{AE}) + S(\rho_{BE}) - S(\rho_E ) - S(\rho_{ABE})
\] 
where $S(\sigma) = -\text{tr}\lb\sigma\log_2(\sigma)\rb$ denotes the von-Neumann entropy of a quantum state $\sigma$. Given a bipartite quantum state $\rho_{AB}\in\D\lb\C^{d_A}\otimes \C^{d_B}\rb$ a quantum state $\sigma_{ABE}\in\D\lb\C^{d_A}\otimes \C^{d_B}\otimes \C^{d_E}\rb$ is called an extension of $\rho_{AB}$ iff $\sigma_{AB} = \rho_{AB}$. For a bipartite quantum state $\rho_{AB}\in\D\lb\C^{d_A}\otimes \C^{d_B}\rb$ the squashed entanglement~\cite{christandl2004squashed,tucci2002entanglement} (w.r.t. the bipartition $A:B$) is defined as
\begin{equation}
E^{A:B}_\text{sq}\lb\rho_{AB}\rb = \frac{1}{2}\inf\lset I(A;B|E)_{\sigma_{ABE}}~:~\sigma_{ABE}\in\D\lb\C^{d_A}\otimes \C^{d_B}\otimes \C^{d_E}\rb \text{ extension of }\rho_{AB} \rset
\label{equ:squash}
\end{equation}
where the dimension $d_E\in\N$ is arbitrary. Now the following bound holds:

\begin{thm}[Squashed entanglement of a quantum channel~\cite{takeoka2014squashed}]
For any quantum channel $T:\M_{d_A}\ra\M_{d_B}$ we have
\[
\mathcal{Q}_\leftrightarrow(T)\leq \mathcal{P}_\leftrightarrow(T)\leq E_{\text{sq}}\lb T\rb
\]
where $E_{\text{sq}}\lb T\rb = \sup\lset E^{A':B}_{\text{sq}}\lb T^{A\ra B}\lb\rho_{A'A}\rb\rb~:~\rho_{A'A}\in\mathcal{D}\lb\C^{d_{A'}}\otimes \C^{d_{A}}\rb\rset$.
\label{thm:squashBound}
\end{thm}

To our knowledge it is currently not known, whether $E_{\text{sq}}\lb T\rb$ is a strong-converse bound on either $\mathcal{Q}_\leftrightarrow(T)$ or $\mathcal{P}_\leftrightarrow(T)$.

Finally, another bound is based on the entanglement cost of a quantum channel~\cite{berta2013entanglement}. For a bipartite quantum state $\rho_{AB}\in\D\lb\C^{d_A}\otimes \C^{d_B}\rb$ the entanglement of formation is defined as
\[
E^{A:B}_{F}\lb\rho_{AB}\rb = \sup_{\lset p_i, \ket{\psi_i}_{AB}\rset} \sum_i p_i S\lb\psi^A_i\rb
\] 
where the supremum is over all pure state decompositions $\rho_{AB} = \sum_i p_i \proj{\psi_i}{\psi_i}_{AB}$ and $\psi^A_i = \text{tr}_{B}\lb\proj{\psi}{\psi}_{AB}\rb$ denotes the reduced density matrix. The entanglement of formation of a quantum channel $T:\M_{d_A}\ra\M_{d_B}$ is defined as 
\[
E_F\lb T\rb = \sup_{\rho_{A'A}} E^{A':B}_{F}\lb T^{A\ra B}\lb\rho_{A'A}\rb\rb
\]
where the supremum is over bipartite states $\rho_{A'A}\in\D\lb\C^{d_{A'}}\otimes \C^{d_A}\rb$ for any dimension $d_{A'}$. The following theorem has essentially been proven in \cite{berta2013entanglement}:

\begin{thm}[Entanglement cost of a quantum channel]
For any quantum channel $T:\M_{d_A}\ra\M_{d_B}$ the entanglement cost of $T$ defined as 
\begin{equation}
E_{C}\lb T\rb = \lim_{n\ra\infty}\frac{1}{n} E_F\lb T^{\otimes n}\rb \leq E_F\lb T\rb
\label{equ:ChanEntCost}
\end{equation}
is a strong-converse bound on $\mathcal{P}_{\leftrightarrow}\lb T\rb$.
\label{thm:ECostBound}
\end{thm}

It has been shown in \cite[Proposition 5]{christandl2004squashed} that 
\[
E_{sq}\lb\rho_{AB}\rb\leq E_F\lb\rho_{AB}\rb
\]
for any bipartite quantum state $\rho_{AB}\in\D\lb \C^{d_A}\otimes \C^{d_B}\rb$. Therefore, it follows from Theorem \ref{thm:squashBound} that $E_{C}\lb T\rb$ is an upper bound on $\mathcal{P}_\leftrightarrow\lb T\rb$. That $E_C\lb T\rb$ is a strong-converse bound on $\mathcal{Q}_\leftrightarrow\lb T\rb$ has been shown in \cite[Theorem 24]{berta2013entanglement}. The inequality in the previous theorem is \cite[Lemma 14]{berta2013entanglement}. The fact that $E_C\lb T\rb$ is a strong-converse bound on $\mathcal{P}_\leftrightarrow\lb T\rb$ has \emph{not} been shown before, but follows easily from \cite{berta2013entanglement}. Specifically, the proof of \cite[Theorem 24]{berta2013entanglement} generalizes to the private capacity (possibly with modified error bounds) by simply using that $\mathcal{P}_\leftrightarrow\lb\ident_2\rb =1$ is a strong-converse capacity\footnote{this follows e.g.~from \cite[Proposition 18]{wilde2016converse} as $\ident_2$ is an erasure channel with erasure probability $0$.} instead of using \cite[Corollary 22]{berta2013entanglement} in the original proof.  

\section{The data-processed triangle inequality}
\label{sec:dataProc}

To establish new strong-converse bounds on the quantity $P_\leftrightarrow$ we need the following inequality for the sandwiched $\alpha$-R\'{e}nyi divergences (see Definition \ref{equ:Renyi}). We call this inequality the data-processed triangle inequality as it resembles a triangle inequality (although with changing distance measure) where some of the involved states are sent through a positive trace-preserving map.   

\begin{thm}[Data-processed triangle inequality]
Let $P:\M_{d_A}\ra\M_{d_B}$ be a positive and trace-preserving map. For any $\alpha \geq 1$ and any quantum states $\rho,\sigma'\in\D\lb\C^{d_A}\rb$ and $\sigma\in\D\lb\C^{d_B}\rb$ we have
\[
D_{\alpha}\lb P(\rho)\|\sigma\rb\leq D_{\alpha}(\rho\|\sigma') + D_{\max}(P(\sigma')\|\sigma).
\]
\label{thm:FundamentalLemma}
\end{thm}
\begin{proof}
Note that there is nothing to show whenever $\text{supp}\lb\rho\rb\nsubseteq \text{supp}\lb\sigma'\rb$ or $\text{supp}\lb P(\sigma')\rb\nsubseteq \text{supp}\lb\sigma\rb$. If $\text{supp}\lb\rho\rb\subseteq \text{supp}\lb\sigma'\rb$ holds true, then positivity of $P$ implies $\text{supp}\lb P(\rho)\rb\subseteq \text{supp}\lb P(\sigma')\rb$. Hence, $\text{supp}\lb P(\sigma')\rb\nsubseteq \text{supp}\lb\sigma\rb$ has to hold whenever both $\text{supp}\lb P(\rho)\rb\nsubseteq \text{supp}\lb\sigma\rb$ and $\text{supp}\lb\rho\rb\subseteq \text{supp}\lb\sigma'\rb$ are fulfilled. We can, therefore, restrict the proof to the cases where all the divergences in the inequality are finite, and w.l.o.g. to the case of full-rank $\sigma$ and $\sigma'$.  

Let $\rho,\sigma'\in\D\lb\C^{d_A}\rb$ and $\sigma\in\D\lb\C^{d_B}\rb$ be fixed quantum states with $\sigma$ and $\sigma'$ of full rank, and $P:\M_{d_A}\ra\M_{d_B}$ a positive trace-preserving map. Consider some fixed $\alpha >1$. By the definition of the $(\alpha,\sigma')\ra (\alpha,\sigma)$-norm (see \eqref{equ:ptoqNorms}) we have
\begin{equation}
\|\Gamma^{-1}_\sigma\circ P(\rho)\|_{(\alpha,\sigma)}\leq \|\Gamma^{-1}_\sigma\circ P\circ\Gamma_{\sigma'}\|_{(\alpha,\sigma')\ra (\alpha,\sigma)}\|\Gamma^{-1}_{\sigma'}(\rho)\|_{(\alpha,\sigma')}.
\label{equ:Proof1NormsBasicnn}
\end{equation}
Applying Theorem \ref{thm:RieszT} for $p_0 = q_0 = 1$ and $p_1 = q_1 = \infty$ and $\theta = \frac{1}{\alpha}$ gives
\begin{equation}
\|\Gamma^{-1}_\sigma\circ P\circ\Gamma_{\sigma'}\|_{(\alpha,\sigma')\ra (\alpha,\sigma)} \leq \|\Gamma^{-1}_\sigma\circ P\circ\Gamma_{\sigma'}\|^{\frac{1}{\alpha}}_{(1,\sigma')\ra (1,\sigma)}\|\Gamma^{-1}_\sigma\circ P\circ\Gamma_{\sigma'}\|^{1-\frac{1}{\alpha}}_{(\infty,\sigma')\ra (\infty,\sigma)} .
\label{equ:Proof1Riesznn}
\end{equation}
For any positive trace-preserving map we have
\begin{equation}
\|\Gamma^{-1}_\sigma\circ P\circ\Gamma_{\sigma'}\|_{(1,\sigma')\ra (1,\sigma)} = \|P\|_{1\ra 1} = 1
\label{equ:Proof11-normnn}
\end{equation}
and, using the Russo-Dye theorem~\cite[Corollary 2.9]{paulsen2002}, we have 
\begin{align}
\|\Gamma^{-1}_\sigma\circ P\circ\Gamma_{\sigma'}\|_{(\infty,\sigma')\ra (\infty,\sigma)} &= \|\Gamma^{-1}_\sigma\circ P\circ\Gamma_{\sigma'}\|_{\infty\ra \infty} \\
&= \|\Gamma^{-1}_\sigma\circ P\circ\Gamma_{\sigma'}(\mathbbm{1}_{d_A})\|_{\infty} = \|\Gamma^{-1}_\sigma\circ P(\sigma')\|_\infty .
\label{equ:Proof1inftynormnn}
\end{align}
Combining equations \eqref{equ:Proof1NormsBasicnn}, \eqref{equ:Proof1Riesznn},\eqref{equ:Proof11-normnn} and \eqref{equ:Proof1inftynormnn} we obtain
\[
\|\Gamma^{-1}_\sigma\circ P(\rho)\|_{(\alpha,\sigma)}\leq \|\Gamma^{-1}_\sigma\circ P(\sigma')\|^{1-\frac{1}{\alpha}}_\infty\|\Gamma^{-1}_{\sigma'}(\rho)\|_{(\alpha,\sigma')}
\]
Taking logarithms, dividing by $1 - \frac{1}{\alpha}$ and writing the resulting inequality in terms of the sandwiched $\alpha$-R\'{e}nyi divergence (see \eqref{equ:SandwichNormVersion}) finishes the proof for $\alpha >1$. Taking the limit $\alpha\ra 1$ gives the statement for $D_1 = D$.

\end{proof}

Note that Theorem \ref{thm:FundamentalLemma} contains some well-known inequalities as special cases. Setting $\sigma = P(\sigma')$ gives the data processing inequality \eqref{equ:DataProc} for the trace-preserving positive map $P$. In the case where $P=\ident_d$ the identity map we get the inequality
\begin{equation}
D_{\alpha}\lb \rho\|\sigma\rb\leq D_{\alpha}(\rho\|\sigma') + D_{\max}(\sigma'\|\sigma)
\label{equ:triangle}
\end{equation}
for any quantum states $\rho,\sigma,\sigma'\in\D\lb\C^d\rb$, which for $\alpha = \infty$ resembles a triangle inequality for $D_{\max}$ (which can easily be shown directly).

We will now apply Theorem \ref{thm:FundamentalLemma} to prove a statement quantifying how the $\alpha$-relative entropy of entanglement changes under partial application of a completely positive map:

\begin{thm}
For any quantum channel $T:\M_{d_{A'}}\ra\M_{d_{B'}}$ and any quantum state $\rho_{AA'B}\in\D(\C^{d_A}\otimes \C^{d_{A'}}\otimes \C^{d_B})$ for $d_A ,d_B\in\N$ arbitrary, we have
\[
E^{A:BB'}_{\alpha}\lb T^{A'\ra B'}(\rho_{AA'B})\rb\leq E_{\max}\lb T\rb + E^{AA':B}_\alpha(\rho_{AA'B}) 
\]
for any $\alpha \geq 1$.
\label{thm:Shifting}
\end{thm}

\begin{proof}
Let $T:\M_{d_{A'}}\ra\M_{d_{B'}}$ be a quantum channel and $\rho_{AA'B}\in\D(\C^{d_A}\otimes \C^{d_{A'}}\otimes \C^{d_B})$ a quantum state for some $d_A ,d_B\in\N$. Furthermore let $\alpha > 1$. For any $\sigma_{AB'B}\in\D(\C^{d_A}\otimes \C^{d_{B'}}\otimes \C^{d_B})$ and $\sigma_{AA'B}'\in\D(\C^{d_A}\otimes \C^{d_{A'}}\otimes \C^{d_B})$ an application of Theorem \ref{thm:FundamentalLemma} for $P = T^{A'\ra B'}$ leads to 
\[
D_\alpha\lb T^{A'\ra B'}(\rho_{AA'B})\|\sigma_{AB'B}\rb\leq D_{\max}\lb T^{A'\ra B'}(\sigma_{AA'B}')\|\sigma_{AB'B}\rb + D_{\alpha}\lb\rho_{AA'B}\|\sigma_{AA'B}'\rb .
\]
Minimizing over $\sigma_{AB'B}\in\text{Sep}_{A:B'B}\lb \C^{d_A}\otimes \C^{d_{B'}}\otimes \C^{d_B}\rb$ and restricting to states $\sigma_{AA'B}'\in\text{Sep}_{AA':B}\lb \C^{d_A}\otimes \C^{d_{A'}}\otimes \C^{d_B}\rb$ leads to 
\begin{align*}
E^{A:B'B}_\alpha\lb T^{A'\ra B'}(\rho_{AA'B})\rb\leq  E^{A:B'B}_{\max}\lb T^{A'\ra B'}(\sigma_{AA'B}')\rb + D_{\alpha}\lb\rho_{AA'B}\|\sigma_{AA'B}'\rb \\
\leq \max_{\tilde{\sigma}_{AA'B}\in\text{Sep}(AA':B)} E^{A:B'B}_{\max}\lb T^{A'\ra B'}(\tilde{\sigma}_{AA'B})\rb + D_{\alpha}\lb\rho_{AA'B}\|\sigma_{AA'B}'\rb.
\end{align*}
Now minimizing over $\sigma_{AA'B}'\in\text{Sep}_{AA':B}\lb \C^{d_A}\otimes \C^{d_{A'}}\otimes \C^{d_B}\rb$ yields
\[
E^{A:B'B}_\alpha\lb T^{A'\ra B'}(\rho_{AA'B})\rb \leq \max_{\tilde{\sigma}_{AA'B}\in\text{Sep}(AA':B)} E^{A:B'B}_{\max}\lb T^{A'\ra B'}(\tilde{\sigma}_{AA'B})\rb + E^{AA':B}_\alpha\lb\rho_{AA'B}\rb .
\]

Let $\sigma_0'\in\text{Sep}_{AA':B}\lb \C^{d_A}\otimes \C^{d_{A'}}\otimes \C^{d_B}\rb$ attain the maximum in the previous equation. As $\sigma_0'$ is separable there is a decomposition of the form $\sigma_0' = \sum^k_{i=1}p_i \gamma^i_{AA'}\otimes \phi^i_{B}$ with $k\in\N$, probabilities $p_i\in\lbr 0,1\rbr$ such that $\sum^k_{i=1} p_i = 1$, and states $\lset\gamma^i_{AA'}\rset^k_{i=1}\subset \D(\C^{d_A}\otimes \C^{d_{A'}})$ and $\lset\phi^i_{B}\rset^k_{i=1}\subset \D(\C^{d_B})$. Now let $\tau^i_{AB'} \in\text{Sep}_{A:B'}\lb \C^{d_A}\otimes \C^{d_{B'}}\rb$ be such that 
\[E^{A:B'}_{\max}\lb T^{A'\ra B'}(\gamma^i_{AA'})\rb = D_{\max}\lb T^{A'\ra B'}(\gamma^i_{AA'})\| \tau^i_{AB'}\rb\]
for each $i\in\lset 1,\ldots ,k\rset$. With these states we get
\begin{align*}
\max_{\tilde{\sigma}_{AA'B}\in\text{Sep}(AA':B)}& E^{A:B'B}_{\max}\lb T^{A'\ra B'}(\tilde{\sigma}_{AA'B})\rb = \min_{\sigma_{AB'B}\in \text{Sep}(A:B'B)} D_{\max}\lb T^{A'\ra B'}(\sigma_0')\| \sigma_{AB'B}\rb \\
& \leq D_{\max}\lb \sum^k_{i=1} p_i T^{A'\ra B'}(\gamma^i_{AA'})\otimes \phi^i_{B}\| \sum^k_{i=1} p_i \tau^i_{AB'}\otimes \phi^i_{B}\rb \\
& \leq \max_{i\in\lset 1,\ldots , k\rset} D_{\max}\lb T^{A'\ra B'}(\gamma^i_{AA'})\otimes \phi^i_{B}\| \tau^i_{AB'}\otimes \phi^i_{B}\rb \\ 
& = \max_{i\in\lset 1,\ldots , k\rset} D_{\max}\lb T^{A'\ra B'}(\gamma^i_{AA'})\| \tau^i_{AB'}\rb \\
& = \max_{i\in\lset 1,\ldots , k\rset} E^{A:B'}_{\max}\lb T^{A'\ra B'}(\gamma^i_{AA'})\rb \\
&\leq E_{\max}\lb T\rb.
\end{align*}  
In the second line of the above computation we used that $\sum^k_{i=1} p_i \tau^i_{AB'}\otimes \phi^i_{B}\in \text{Sep}_{A:B'B}\lb\C^{d_A}\otimes \C^{d_{B'}}\otimes\C^{d_B}\rb$ as $\tau^i_{AB'}\in\text{Sep}_{A:B'}\lb\C^{d_A}\otimes \C^{d_{B'}}\rb$ by definition. In the third line we used that $D_{\max}$ is joint quasi-convex \cite[Lemma 9]{datta2009min} and in the fourth line that $D_{\max}\lb\rho_1\otimes \sigma_1\|\rho_2\otimes \sigma_2\rb = D_{\max}\lb\rho_1\|\rho_2\rb + D_{\max}\lb\sigma_1\|\sigma_2\rb$ for any quantum states $\rho_1,\rho_2\in\D\lb\C^d\rb$ and $\sigma_1,\sigma_2\in\D\lb\C^{d'}\rb$ with $\text{supp}\lbr\rho_1\rbr\subseteq \text{supp}\lbr\rho_2\rbr$ and $\text{supp}\lbr\sigma_1\rbr\subseteq \text{supp}\lbr\sigma_2\rbr$ (see for instance \cite[Theorem 2]{muller2013quantum}).

\end{proof}

The following corollary bounds the $\alpha$-relative entropy of entanglement of the state obtained from alternately applying an LOCC-operation and a partial quantum channel to some tripartite initial state. 

\begin{cor}
Let $T:\M_{d_{A'}}\ra\M_{d_{B'}}$ be a quantum channel, $m\in\N$ ,and $d_{A_i},d_{B_i}\in\N$ dimensions for each $i\in\lset 1,\ldots ,m\rset$. Consider LOCC-operations $\lset L_i\rset^{m}_{i=1}$ w.r.t. the bipartition into $A$ and $B$ systems acting as 
\[
L_{i}:\M_{d_{A_i}}\otimes \M_{d_{B'}}\otimes \M_{d_{B_i}}\ra \M_{d_{A_{i+1}}}\otimes\M_{d_{A'}}\otimes \M_{d_{B_i}} 
\]
for any $i\in\lset 1,\ldots ,m\rset$ and 
\[
L_{m}:\M_{d_{A_{m}}}\otimes \M_{d_{B'}}\otimes \M_{d_{B_m}}\ra \M_{d_{\tilde{A}}}\otimes \M_{d_{\tilde{B}}}
\]
for arbitrary $d_{\tilde{A}},d_{\tilde{B}}\in\N$. For any quantum state $\rho_{A_1A'B_1}\in\D\lb\C^{d_{A_1}}\otimes \C^{d_{A'}}\otimes \C^{d_{B_1}}\rb$ consider the state 
\[
\phi_{\tilde{A}\tilde{B}} = L_{m}\circ \prod^{m-1}_{i=1}\lb T^{A'\ra B'}\circ L_i\rb \circ T^{A'\ra B'}\lb\rho_{A_1 A' B_1}\rb .
\]
Then we have
\[
E_\alpha^{\tilde{A}:\tilde{B}}\lb \phi_{\tilde{A}\tilde{B}}\rb \leq m E_{\max}\lb T\rb + E_\alpha^{A_1 A': B_1}\lb\rho_{A_1 A' B_1}\rb
\]  
for any $\alpha >1$.
\label{Lem:UpperBoundEntPhi}
\end{cor}

\begin{proof}
For $i\in\lset 1,\ldots ,m\rset$ we define the states 
\[
\sigma^{(i)}_{A_i B' B_i} = \prod^{i-1}_{k=1}\lb T^{A'\ra B'}\circ L_i\rb \circ T^{A'\ra B'}\lb\rho_{A_1 A' B_1}\rb. 
\]
By the data-processing inequality \eqref{equ:DataProc} it is easy to see that $E_\alpha$ is non-increasing under LOCC-operations as such operations preserve the set of separable states. Using this fact and Theorem \ref{thm:Shifting} alternately gives  
\begin{align*}
E_\alpha^{\tilde{A}:\tilde{B}}\lb \phi_{\tilde{A}\tilde{B}}\rb &= E_\alpha^{\tilde{A}:\tilde{B}}\lb L_{m}\lb \sigma_{A_mB'B_m}^{(m)}\rb\rb \\
&\leq E_\alpha^{A_m:B'B_m}\lb \sigma_{A_mB'B_m}^{(m)}\rb \\
&\leq E_{\max}\lb T\rb + E_\alpha^{A_{m-1}:B'B_{m-1}}\lb \sigma_{A_{m-1}B'B_{m-1}}^{(m-1)}\rb \\
&\vdots \\
&\leq (m-1) E_{\max}\lb T\rb + E^{A_1 B' B_1}_\alpha\lb T^{A'\ra B'}\lb\rho_{A_1 A' B_1}\rb\rb \\
&\leq m E_{\max}\lb T\rb + E^{A_1 A' B_1}_\alpha\lb \rho_{A_1 A' B_1}\rb.
\end{align*}

\end{proof}

In the next section we will apply the previous corollary to the output state of a protocol for $\mathcal{P}_\leftrightarrow$. This will establish the strong-converse bound in terms of $E_{\max}$.

\section{Strong converse bound on $\mathcal{P}_\leftrightarrow$}
\label{sec:StrongConverse}

To prove a strong converse bound on  $\mathcal{P}_\leftrightarrow$ we will use some notions introduced in~\cite{wilde2016converse}. Consider a private state  
\[
\gamma_{A_kB_kA_sB_s} = U^{\text{tw}}_{A_kB_kA_sB_s}\lb\omega_{A_kB_k}\otimes \sigma_{A_sB_s}\rb (U^{\text{tw}}_{A_kB_kA_sB_s})^\dagger 
\]
where $U^{\text{tw}}_{A_kB_kA_sB_s}$ denotes a twisting unitary (see Definition \ref{defn:PrivState}). A privacy test corresponding to $\gamma_{A_kB_kA_sB_s}$ (see~\cite[Definition 6]{wilde2016converse}) is a 2-outcome measurement given by the POVM
\begin{equation}
\lset \Pi_{A_kB_kA_sB_s} , \one_{A_kB_kA_sB_s} -\Pi_{A_kB_kA_sB_s} \rset
\label{equ:PrivTest}
\end{equation}
for the projector $\Pi_{A_kB_kA_sB_s} = U^{\text{tw}}_{A_kB_kA_sB_s}\lb\omega_{A_kB_k}\otimes \one_{A_sB_s}\rb \lb U^{\text{tw}}_{A_kB_kA_sB_s}\rb^\dagger$. It can be shown that a separable state only has a low probability of passing a privacy test (i.e.\ the measurement \eqref{equ:PrivTest} giving the outcome corresponding to $\Pi_{A_kB_kA_sB_s}$) corresponding to a private state. More specifically (see \cite[equation (281)]{horodecki2009general} or \cite[Lemma 8]{wilde2016converse}) for a privacy test \eqref{equ:PrivTest} corresponding to a private state $\gamma_{A_kB_kA_sB_s}$ with $K$-dimensional key part (i.e.\ $d_{A_k}=d_{B_k}=K$ as in Definition \ref{defn:PrivState}) it holds that
\begin{equation}
\text{tr}\lb \Pi_{A_kB_kA_sB_s}\sigma_{A_kB_kA_sB_s}\rb\leq \frac{1}{K}
\label{equ:PrivTestSep}
\end{equation}
for any separable state $\sigma_{A_kB_kA_sB_s}\in\text{Sep}_{A_kA_s:B_kB_s}\lb\C^{d_{A_k} d_{A_s}}\otimes \C^{d_{B_k} d_{B_s}}\rb$. At the same time the probability of a state passing the privacy test can be related to its distance to the private state. The following Lemma has been shown in~\cite{wilde2016converse}:

\begin{lem}[Lemma 7 in \cite{wilde2016converse}]
Let $\rho_{A_kB_kA_sB_s}\in\D\lb\C^{d_{A_k}}\otimes \C^{d_{B_k}}\otimes \C^{d_{A_s}}\otimes \C^{d_{B_s}}\rb$ be a quantum state and $\Pi_{A_kB_kA_sB_s}$ the projector appearing in the privacy test (see \eqref{equ:PrivTest}) corresponding to a private state $\gamma_{A_kB_kA_sB_s}$. Then 
\[
\text{tr}\lb \Pi_{A_kB_kA_sB_s}\rho_{A_kB_kA_sB_s}\rb \geq F(\rho_{A_kB_kA_sB_s},\gamma_{A_kB_kA_sB_s} )
\] 
where $F(\rho,\sigma) = \|\sqrt{\rho}\sqrt{\sigma}\|^2_1$ denotes the fidelity.
\label{lem:WildeLem}
\end{lem}

Now we can prove a bound on the error of private state generation for protocols assisted by classical communication. The proof follows a method given in \cite{wilde2014strong} and uses ideas from \cite{wilde2016converse}.

\begin{lem}[Bound on private communication error]
Let $T:\M_{d_{A'}}\ra\M_{d_{B'}}$ be a quantum channel and $\alpha\in (1,\infty)$. For any $k,m\in\N$ the error $\epsilon >0$ in an $(k,m,\epsilon)$-coding scheme for private state generation assisted by classical communication (as in Definition \ref{defn:PrivCodingScheme}) fulfills
\[
\epsilon \geq 1-2^{-\frac{\alpha-1}{2\alpha}\lb k - m E_{\max}(T)\rb}.
\]
\label{Lem:BoundComErrPriv}
\end{lem}
\begin{proof}
Let $\phi^{(m)}_{\tilde{A}\tilde{B}}\in\D\lb\C^{d_{\tilde{A}}}\otimes \C^{d_{\tilde{B}}}\rb$ denote the output state of the $(k,m,\epsilon)$-coding scheme, i.e. 
\[
\phi^{(m)}_{\tilde{A}\tilde{B}} = L_m\circ\prod^{m-1}_{i=1} \lb T^{A'\ra B'}\circ L_i\rb\circ T^{A'\ra B'}\lb\rho^{(1)}_{A_1 A' B_1}\rb
\]
for LOCC-operations $\lset L_i\rset^{m}_{i=1}$ and initial state $\rho_{A_1 A' B_1}\in \text{Sep}_{A_1 A':B_1}\lb\C^{d_{A_1}}\otimes \C^{d_{A'}}\otimes \C^{d_{B_1}}\rb$ as in Definition \ref{defn:CodingScheme}. Note that by the form of $\phi^{(m)}_{\tilde{A}\tilde{B}}$ we can apply Lemma \ref{Lem:UpperBoundEntPhi} to show
\begin{equation}
E^{\tilde{A}:\tilde{B}}_\alpha\lb\phi^{(m)}_{\tilde{A}\tilde{B}}\rb \leq m E_{\max}\lb T\rb 
\label{equ:phiState}
\end{equation}
where we used that by separability $E^{A_1A':B_1}_\alpha\lb\rho^{(1)}_{A_1 A' B_1}\rb = 0$.

By assumption (from Definition \ref{defn:PrivCodingScheme}) we have $d_{\tilde{A}} = d_{A_k}d_{A_s}$ and $d_{\tilde{B}} = d_{B_k}d_{B_s}$ with $d_{A_k}=d_{B_k}=2^k$ and $d_{A_s}=d_{B_s}$, and there exists a private state $\gamma_{A_kB_kA_sB_s}\in\D\lb\C^{d_{\tilde{A}}}\otimes \C^{d_{\tilde{B}}}\rb$ with $2^k$-dimensional key part (see Definition \ref{defn:PrivState}) such that 
\begin{equation}
\epsilon = \frac{1}{2}\|\phi^{(m)}_{\tilde{A}\tilde{B}} - \gamma_{A_kB_kA_sB_s}\|_1.
\label{equ:epsilon}
\end{equation}
Let $\Pi^{(k)}_{A_kB_kA_sB_s}$ denote the projector in the privacy test corresponding to $\gamma_{A_kB_kA_sB_s}$ (see \eqref{equ:PrivTest}). Then by Lemma \ref{lem:WildeLem} we have 
\begin{equation}
F := \text{tr}\lb\Pi^{(k)}_{A_kB_kA_sB_s} \phi^{(m)}_{\tilde{A}\tilde{B}}\rb \geq F(\phi^{(m)}_{\tilde{A}\tilde{B}},\gamma_{2^k})
\label{equ:F}
\end{equation}
Now we define a binary flag channel $B:\D\lb\C^{d_{\tilde{A}}}\otimes \C^{d_{\tilde{B}}}\rb\ra \M_2$ by 
\[
B(X) = \text{tr}\lb \Pi^{(k)}_{A_kB_kA_sB_s}X\rb\ketbra{1} + \text{tr}\lb (\one_{A_kB_kA_sB_s} - \Pi^{(k)}_{A_kB_kA_sB_s}) X\rb\ketbra{0}
\]
where $\ket{0},\ket{1}\in\C^2$ denote the computational basis states. For any separable state $\sigma_{\tilde{A}\tilde{B}}\in \text{Sep}_{\tilde{A}:\tilde{B}}\lb \C^{d_{\tilde{A}}}\otimes \C^{d_{\tilde{B}}}\rb$ we can compute
\begin{align*}
D_\alpha\lb \phi^{(m)}_{\tilde{A}\tilde{B}}\| \sigma_{\tilde{A}\tilde{B}}\rb &\geq D_\alpha\lb B\lb\phi^{(m)}_{\tilde{A}\tilde{B}}\rb\| B\lb\sigma_{\tilde{A}\tilde{B}}\rb\rb \\
&= \frac{1}{\alpha -1}\log_2\lb F^\alpha p^{1-\alpha} + (1-F)^\alpha(1-p)^{1-\alpha}\rb \\
&\geq \frac{1}{\alpha -1}\log_2\lb F^\alpha p^{1-\alpha}\rb \\
&\geq \frac{1}{\alpha-1}\log_2\lb F^\alpha \lb\frac{1}{2^k}\rb^{1-\alpha}\rb \\
&= \frac{\alpha}{\alpha - 1}\log_2(F) + k.
\end{align*}
Here we introduced $p = \text{tr}\lb\sigma_{\tilde{A}\tilde{B}}\Pi^{(k)}_{\tilde{A}\tilde{B}}\rb$ and used the data-processing inequality \eqref{equ:DataProc} of the sandwiched $\alpha$-R\'{e}nyi divergences for the first inequality. In the last inequality we used $p\leq \frac{1}{2^k}$ which follows from \eqref{equ:PrivTestSep} and separability of $\sigma_{\tilde{A}\tilde{B}}$. 
 
Minimizing over all separable states $\sigma_{\tilde{A}\tilde{B}}\in \text{Sep}_{\tilde{A}:\tilde{B}}\lb \C^{d_{\tilde{A}}}\otimes \C^{d_{\tilde{B}}}\rb$ on the left-hand-side of the previous equation and using \eqref{equ:phiState} gives
\[
m E_{\max}\lb T\rb \geq E_\alpha^{\tilde{A}:\tilde{B}}\lb \phi^{(m)}_{\tilde{A}\tilde{B}}\rb \geq \frac{\alpha}{\alpha - 1}\log_2(F) + k .
\] 
By applying the Fuchs-van-de-Graaf inequality~\cite[Theorem 1]{fuchs1999cryptographic} and \eqref{equ:F} we get 
\[
\epsilon \geq 1-\sqrt{F(\phi^{(m)}_{\tilde{A}\tilde{B}},\gamma_{2^k})} \geq 1-\sqrt{F} \geq 1 - 2^{-\frac{\alpha-1}{2\alpha}\lb k - m E_{\max}\lb T\rb\rb}
\]
for the communication error $\epsilon$ from \eqref{equ:epsilon}.

\end{proof}

\begin{thm}
Let $T:\M_{d_{A'}}\ra\M_{d_{B'}}$ be a quantum channel. Then the quantity $E_{\max}\lb T\rb$ is a strong-converse bound on $\mathcal{P}_\leftrightarrow(T)$.
\label{thm:SCBound}
\end{thm}
\begin{proof}
Consider $R>E_{\max}\lb T\rb$ such that for each $\nu\in\N$ there exists an $(k_\nu,m_\nu,\varepsilon_\nu)$-coding scheme for private state generation assisted by classical communication (as in Definition \ref{defn:CodingSchemeEG}) with $m_\nu\ra\infty$ as $\nu\ra\infty$ and $R=\lim_{\nu\to\infty}\frac{k_\nu}{m_\nu}$. There exists a $\delta >0$ and a $\nu_0\in\N$ such that 
\[
\frac{k_\nu}{m_\nu} > E_{\max}\lb T\rb + \delta
\] 
for all $\nu\geq \nu_0$. Therefore, using Lemma \ref{Lem:BoundComErrPriv} we have for any $\nu\geq \nu_0$ and $\alpha >1$ that
\begin{align*}
\epsilon_\nu &\geq 1 - 2^{-\frac{\alpha-1}{2\alpha}\lb k_\nu - m_\nu E_{\max}\lb T\rb\rb} \\
&\geq 1 - 2^{-\frac{\alpha-1}{2\alpha} m_\nu \delta} \ra 1
\end{align*} 
as $\nu\ra\infty$. 
\end{proof}

Finally we can regularize the above bound. Consider the regularized max-relative entropy of a quantum channel $T:\M_{d_{A'}}\ra\M_{d_{B'}}$ 
\[
E^\infty_{\max}\lb T\rb := \lim_{n\ra \infty}\frac{1}{n} E_{\max}\lb T^{\otimes n}\rb.
\]
As a special case of \cite[Theorem 13]{wilde2016converse} (which can also be shown directly following the proof of \cite[Theorem 6]{7282883} for the quantity $E_{\max}$) we have for any $n\in\N$
\[
E_{\max}\lb T^{\otimes n}\rb \leq n E_{\max}\lb T\rb + d_{A'}\log_2(n).
\]
Dividing by $n$ and taking the limit $n\ra \infty$ implies
\[
E^\infty_{\max}\lb T\rb \leq E_{\max}\lb T\rb.
\]
We can therefore improve the bound from Theorem \ref{thm:SCBound} (which is in particular an upper bound on $\mathcal{P}_\leftrightarrow$) by regularization. Note that by Definition \ref{defn:PrivCap} we have $\mathcal{P}_{\leftrightarrow}\lb T\rb = \lim_{n\ra\infty}\frac{1}{n}\mathcal{P}_{\leftrightarrow}\lb T^{\otimes n}\rb$. Applying the bound from Theorem \ref{thm:SCBound} for the channels $T^{\otimes n}$ and noting that $\mathcal{Q}_\leftrightarrow \leq \mathcal{P}_\leftrightarrow$ (by Definition \ref{defn:PrivCap}) implies:
 
\begin{cor}[Regularized upper bound on $\mathcal{P}_\leftrightarrow$]
For any quantum channel $T:\M_{d_{A'}}\ra\M_{d_{B'}}$ we have
\[
\mathcal{Q}_\leftrightarrow(T)\leq \mathcal{P}_\leftrightarrow(T)\leq E^\infty_{\max}\lb T\rb \leq E_{\max}\lb T\rb.
\]
\label{cor:RegUpperBoundP2}
\end{cor}

\section{Properties of $E_{\max}\lb T\rb$}

\subsection{Non-lockability}
\label{sec:Lock}

An entanglement measure is called non-lockable~\cite{horodecki2005locking} if tracing out a subsystem of dimension $d\in\N$ can only change the measure by an amount logarithmic in $d$. Here we show that this is the case for the max-relative entropy of entanglement (cf. Theorem \ref{thm:NonLockEmax}). As a consequence we show that for a quantum channel $T^{A\ra BC}:\M_{d_A}\ra\M_{d_{B}d_{C}}$ the difference of the quantities $E_{\max}\lb T^{A\ra BC}\rb$ and $E_{\max}\lb \text{tr}_C\circ T^{A\ra BC}\rb$ can be at most logarithmic in $d_C$. We start with an elementary lemma which is probably known:

\begin{lem}
For some $k\in\N$ consider a convex combination $\rho_{AB} = \sum^k_{i=1} p_i \rho_{AB}^i$ of bipartite quantum states $\rho_{AB}^i\in\D\lb\C^{d_A}\otimes \C^{d_B}\rb$ and $p_i\in \lbr 0,1\rbr$ for $i\in\lset 1,\ldots ,k\rset$ such that $\sum^k_{i=1} p_i = 1$. Then we have 
\[
\sum^k_{i=1} p_i E^{A:B}_{\max}\lb\rho_{AB}^i\rb \leq E^{A:B}_{\max}\lb\rho_{AB}\rb + \sum^k_{i=1} p_i D_{\max}\lb\rho_{AB}^i\|\rho_{AB}\rb .
\]
\label{lem:SimpleLem}
\end{lem}

\begin{proof}
Given states $\rho_i\in\D\lb\C^{d_A}\otimes \C^{d_B}\rb$ for $i\in\lset 1,\ldots ,k\rset$ and $\rho = \sum^k_{i=1} p_i \rho_{AB}^i$  with probabilities $\lset p_i\rset^k_{i=1}\subset \lbr 0,1\rbr$ fulfilling $\sum^k_{i=1} p_i = 1$. Note that applying Theorem \ref{thm:FundamentalLemma} for $P = \ident_d$ and $\alpha = \infty$ gives
\[
D_{\max}\lb\rho_{AB}^i\|\sigma_{AB}\rb\leq D_{\max}\lb\rho_{AB}\|\sigma_{AB}\rb + D_{\max}\lb\rho_{AB}^i\|\rho_{AB}\rb
\]
for any states $\sigma_{AB}\in\D\lb\C^{d_A}\otimes \C^{d_B}\rb$. Minimizing over $\sigma_{AB}\in\text{Sep}_{A:B}\lb\C^{d_A}\otimes \C^{d_B}\rb$ leads to
 \[
E^{A:B}_{\max}\lb\rho_{AB}^i\rb\leq E^{A:B}_{\max}\lb\rho_{AB}\rb + D_{\max}\lb\rho_{AB}^i\|\rho\rb.
\]
Finally multiplying the above inequalities by $p_i$ for each $i\in\lset 1,\ldots k\rset$ and summing over $i$ leads to the statement of the lemma. 
\end{proof}

With the previous lemma we can show that the max-relative entropy of entanglement is non-lockable. The argument is similar to an argument given in~\cite{horodecki2005locking} for the relative entropy of entanglement.

\begin{thm}[Non-lockability of the max-relative entropy of entanglement]
For any tripartite state $\rho_{ABB'}\in\D\lb\C^{d_A}\otimes \C^{d_B}\otimes \C^{d_{B'}}\rb$ we have
\[
E^{A:BB'}_{\max}\lb\rho_{ABB'}\rb - E^{A:B}_{\max}\lb\rho_{AB}\rb\leq 2\log_2(d_{B'}).
\]
\label{thm:NonLockEmax}
\end{thm}
\begin{proof}

Note that 
\begin{align*}
\rho_{AB}\otimes \frac{\one_{d_{B'}}}{d_{B'}} &= \int_{\mathcal{U}_{d_{B'}}} \lb\one_{d_{AB}}\otimes U\rb\rho_{ABB'}\lb\one_{d_{AB}}\otimes U\rb^\dagger \text{dU} \\
&=\frac{1}{k}\sum^k_{i=1} \lb\one_{d_{AB}}\otimes U_i\rb\rho_{ABB'}\lb\one_{d_{AB}}\otimes U_i\rb^\dagger
\end{align*}
where the integral is with respect to the Haar-measure on the unitary group $\mathcal{U}_{d_{B'}}$, and where we used unitaries $\lb U_i\rb^k_{i=1}$ forming a unitary $2$-design (see \cite{roy2009unitary}). Applying Lemma \ref{lem:SimpleLem} for the above convex combination gives
\begin{align*}
\sum^k_{i=1} \frac{1}{k} &E^{A:BB'}_{\max}\lb\lb\one_{d_{AB}}\otimes U_i\rb\rho_{ABB'}\lb\one_{d_{AB}}\otimes U_i\rb^\dagger\rb - E^{A:BB'}_{\max}\lb \rho_{AB}\otimes \frac{\one_{d_{B'}}}{d_{B'}}\rb \\
&\leq \sum^k_{i=1} \frac{1}{k} D_{\max}\lb\lb\one_{d_{AB}}\otimes U_i\rb\rho_{ABB'}\lb\one_{d_{AB}}\otimes U_i\rb^\dagger\|\rho_{AB}\otimes \frac{\one_{d_{B'}}}{d_{B'}}\rb \\
& = D_{\max}\lb\rho_{ABB'}\|\rho_{AB}\otimes \frac{\one_{d_{B'}}}{d_{B'}}\rb \\
&\leq 2\log_2\lb d_{B'}\rb.
\end{align*}
Here we used that $D_{\max}$ is invariant under unitary transformations applied to both of its arguments, and that $\rho_{ABB'}\leq d_{B'}\rho_{AB}\otimes \one_{d_{B'}}$, which by \eqref{equ:defDmax} implies the last inequality. Finally note that 
\[
E^{A:BB'}_{\max}\lb \rho_{AB}\otimes \frac{\one_{d_{B'}}}{d_{B'}}\rb = E^{A:B}_{\max}\lb \rho_{AB}\rb
\]
by monotonicity under local operations.

\end{proof}

Finally by applying Theorem \ref{thm:NonLockEmax} to the quantity $E_{\max}$ (see \eqref{equ:maxRelEntrQChan}) we obtain the following:

\begin{cor}
Let $T^{A\ra BC}:\M_{d_A}\ra\M_{d_B}\otimes \M_{d_C}$ be a quantum channel and consider the reduced quantum channel $T^{A\ra B} = \text{tr}_{C}\circ T^{A\ra BC}$. Then 
\[
E_{\max}\lb T^{A\ra BC}\rb\leq 2\log_2(d_C) + E_{\max}\lb T^{A\ra B}\rb
\]
\label{cor:RedChanUp}
\end{cor}

The previous corollary is used in Section \ref{sec:Flower} to show that the bound from Corollary \ref{cor:RegUpperBoundP2} improves on both the transposition bound (see Theorem \ref{thm:transpBound}) and the squashed entanglement bound (see Theorem \ref{thm:squashBound}).

\subsection{Simplified upper bounds}
\label{sec:SimplUpp}

The optimizations over input states and separable states make it hard to compute $E_{\max}\lb T\rb$ (see \eqref{equ:maxRelEntrQChan}) for a concrete quantum channel $T:\M_{d_A}\ra\M_{d_B}$. In the following we will give a slightly simpler bound in terms of the quantity
\[
B_{\max}\lb T\rb = \min\lset D_{\max}\lb C_T\| C_S\rb ~:~ S:\M_{d_A}\ra\M_{d_B}\text{ entanglement breaking quantum channel}\rset.
\]
Here $C_T$ and $C_S$ denote Choi matrices (see \eqref{equ:Choi}) of the channels $T$ and $S$. Recall that a quantum channel $S:\M_{d_A}\ra\M_{d_B}$ is called entanglement breaking~\cite{horodecki2003entanglement} iff $S^{A\ra B}\lb\rho_{A'A}\rb$ is separable for any bipartite state $\rho_{A'A}\in\D\lb\C^{d_{A'}}\otimes \C^{d_A}\rb$ and where $A'$ is a system of any dimension. This is equivalent to separability of the Choi matrix $C_S$. Note that since $\text{tr}_B\lb C_S\rb = \one_{d_{A'}}/d_{A'}$ (where we used that $C_S = S^{A\ra B}\lb\omega_{A'A}\rb$ now for $d_{A'}=d_A$) the above quantity is in general different from $E^{A':B}_{\max}\lb C_T\rb$.

\begin{thm}[Simplified upper bound]

For a quantum channel $T:\M_{d_A}\ra\M_{d_B}$ we have
\[
E_{\max}\lb T\rb\leq B_{\max}\lb T\rb.
\]
\label{thm:SimplerUpperBound}
\end{thm} 

\begin{proof}

Let $S:\M_{d_A}\ra\M_{d_B}$ be an entanglement breaking channel. For any bipartite quantum state $\rho_{A'A}\in\D\lb\C^{d_{A'}\otimes \C^{d_A}}\rb$ the output state $S^{A\ra B}\lb\rho_{A'A}\rb$ is separable. Therefore we have
\begin{align*}
E^{A':B}_{\max}\lb T^{A\ra B}\lb \rho_{A'A}\rb\rb &= \inf\lset D_{\max}\lb T^{A\ra B}\lb\rho_{A'A}\rb\|\sigma_{A'B}\rb ~:~ \sigma_{A'B}\in\text{Sep}_{A':B}\lb\C^{d_{A'}}\otimes \C^{d_B}\rb\rset \\
&\leq D_{\max}\lb T^{A\ra B}\lb\rho_{A'A}\rb\|S^{A\ra B}\lb\rho_{A'A}\rb\rb \\
& = \inf\lset \lambda\geq 0 : T^{A\ra B}\lb\rho_{A'A}\rb \leq 2^{\lambda}S^{A\ra B}\lb\rho_{A'A}\rb\rset.
\end{align*}
The condition in the last infimum is certainly fulfilled if the linear map $2^{\lambda}S - T$ is completely positive (in this case the condition holds for any state $\rho_{A'A}$). Expressing complete positivity of this linear map in terms of the Choi matrix~\cite{choi1975completely} yields 
\[
E^{A':B}_{\max}\lb T^{A\ra B}\lb \rho_{A'A}\rb\rb \leq \inf\lset \lambda\geq 0 : C_T \leq 2^{\lambda}C_S\rset = D_{\max}\lb C_T\| C_S\rb,  
\]
where $C_T = T^{A\ra B}\lb\omega_{A'A}\rb$ denotes the Choi matrix of $T$ (and $C_S$ the Choi matrix of $S$). As the previous bound holds for any input state $\rho_{A'A}$ and any entanglement breaking channel $S:\M_{d_A}\ra\M_{d_B}$ the proof is finished.

\end{proof}

\section{Applications}

\subsection{Flower channels}
\label{sec:Flower}

Here we will compare the bound from Corollary \ref{cor:RegUpperBoundP2} to previously known bounds. Numerical computations show that for the qubit depolarizing channel, the qubit erasure channel, and the qubit amplitude damping channel our bound does not outperform the transposition bound. It should also be noted that for channels implementable from their image (see Definition \ref{defn:ImpleFrIm} in Appendix A) including all teleportation-covariant channels, the bound based on the relative entropy of entanglement (see \cite{pirandola2015ultimate}) performs better than our bound (based on the max-relative entropy of entanglement). However, for many important quantum channels (e.g.~the channels considered in this section and in Section \ref{sec:non-rep}) it is currently not known whether they are implementable from their image. Moreover, in Appendix A we provide an example of a quantum channel which cannot be implemented from its image. Instead of estimating our bound for the commonly used standard examples, we will consider a particular construction of quantum channels in high dimensions. This exploits the non-lockability of our bound to outperform the previously known bounds. As the transposition bound (see Theorem \ref{thm:transpBound}) only upper bounds $\mathcal{Q}_\leftrightarrow$ and not $\mathcal{P}_\leftrightarrow$ we will only consider the former quantity in this section.

Here we will use a particular family of channels (so called flower channels) for which the transposition bound (see Theorem \ref{thm:transpBound}), the bound based on the squashed entanglement (see Theorem \ref{thm:squashBound}), and thereby also the entanglement cost bound (see Theorem \ref{thm:ECostBound}) perform exceptionally badly. The reason of this bad performance is that all these bounds are lockable~\cite{horodecki2005locking}. The new bound based on the max-relative entropy is non-lockable (cf. Corollary \ref{cor:RedChanUp}), which leads to an improvement compared to the other bounds. Moreover, the improvement can be made arbitrarily large by increasing the dimension of the channels.  

For $d\in\N$ consider the so called ``flower'' states given by
\begin{equation}
\rho^f_{AA'BB'} = \frac{1}{2d}\sum^d_{i,k=1}\sum^2_{j,l=1} \bra{k}U^\dagger_{l}U_j\ket{i}~ \proji{ii}{kk}{AB}\otimes\proji{jj}{ll}{A'B'}\in\D\lb\C^{d}\otimes \C^{2}\otimes\C^{d}\otimes \C^2\rb
\label{equ:flowerState}
\end{equation}
where $U_1 = \one_d$ and $U_2$ is the quantum Fourier transformation with entries 
\[
(U_2)_{j,k} = \frac{1}{\sqrt{d}}e^{2\pi ijk/d}
\]
for $j,k\in\lset 1,\ldots ,d\rset$. In \cite{horodecki2005locking} and \cite{christandl2005uncertainty} several entanglement measures have been computed for these states. The squashed entanglement (see \eqref{equ:squash}) is given by (see \cite[Proposition 4]{christandl2005uncertainty})
\begin{equation}
E^{AA':BB'}_{\text{sq}}\lb \rho^f_{AA'BB'}\rb = 1 + \frac{1}{2}\log_2(d),
\label{equ:FlowerSquash}
\end{equation}
and the logarithmic negativity is given by (see \cite[p. 2]{horodecki2005locking}) 
\begin{equation}
\log_2\lb\|(\rho^f_{AA'BB'})^{T_{BB'}}\|_1\rb = \log_2\lb\sqrt{d}+1\rb.
\label{equ:FlowerTrans}
\end{equation}
Note that the previous quantities are unbounded in the limit $d\ra \infty$. However, the actual entanglement in the states $\rho^f_{AA'BB'}$ is small, because tracing out the 2-dimensional system $B'$ leads a separable state 
\[
\rho^f_{AA'B} = \frac{1}{2d}\sum^d_{i=1}\sum^2_{j=1} \proji{ii}{ii}{AB}\otimes \proji{j}{j}{A'} .
\]
 
The two marginals of a flower state $\rho^f_{AA'BB'}$ fulfill $\rho^f_{BB'} = \rho^f_{AA'} = \one_{2d}/(2d)$. Therefore, by the Choi-Jamiolkowski isomorphism~\cite{choi1975completely} there is a unital quantum channel $T^{AA'\ra BB'}_f:\M_{2d}\ra\M_{2d}$ with Choi matrix $\rho^f_{AA'BB'}$. We call this channel a flower channel. Note that the reduced channel $T_f^{AA'\ra B} = \text{tr}_{B'}\circ T_f^{AA'\ra BB'}$ is entanglement breaking as its Choi matrix is the separable state $\rho^f_{AA'B}$. This implies that $E_{\max}\lb T_f^{AA'\ra B}\rb = 0$ and using Corollary \ref{cor:RedChanUp} and the non-regularized bound from Corollary \ref{cor:RegUpperBoundP2} we get 
\begin{equation}
Q_2\lb T_f^{AA'\ra BB'}\rb\leq E_{\max}\lb T_f^{AA'\ra BB'}\rb\leq 2 + E_{\max}\lb T_f^{AA'\ra B}\rb = 2.
\label{equ:FlowerMaxBound}
\end{equation}
We can also estimate the transposition bound (see Theorem \ref{thm:transpBound}) and the bound based on the squashed entanglement (see Theorem \ref{thm:squashBound}). By \eqref{equ:FlowerTrans} and \eqref{equ:FlowerSquash} we have 
\[
\log_2\lb\| \vartheta_{BB'}\circ T_f^{AA'\ra BB'}\|_\diamond\rb \geq \log_2\lb\|(\rho^f_{AA'BB'})^{T_{BB'}}\|_1\rb = \log_2\lb\sqrt{d}+1\rb
\]
\[
E_{sq}\lb T_f^{AA'\ra BB'}\rb \geq E^{AA':BB'}_{\text{sq}}\lb \rho^f_{AA'BB'}\rb = 1 + \frac{1}{2}\log_2(d). \]
These computations show that \eqref{equ:FlowerMaxBound} improves upon the squashed entanglement and by the discussion following Theorem \ref{thm:ECostBound} also upon the entanglement cost bound for $d>2$. For $d>9$ our bound also improves upon the transposition bound. All these improvements can be made arbitrary large by increasing the dimension $d$.

\subsection{Non-repeatable private capacity}
\label{sec:non-rep}

In \cite{bauml2015limitations} a general paradigm has been introduced for sharing key using several quantum states sequentially connecting communication nodes to bridge a possibly long distance between the communicating parties $A$ and $B$. Consider the case where only one intermediate node $C$ connected to $A$ and $B$ by quantum states $\rho^{(1)}_{AC}$ and $\rho^{(2)}_{CB}$ is available. The supremum of rates with which private key can be established between $A$ and $B$ using arbitrary LOCC-operations acting on many copies of the two states is the repeatable key rate $\mathcal{K}_{A\leftrightarrow C\leftrightarrow B}\lb\rho^{(1)}_{AC},\rho^{(2)}_{CB}\rb$ (see \cite{bauml2015limitations}). 

It is clear that in the same scenario any pair of states with distillable entanglement~\cite{bennett1996concentrating} can be used to create entanglement between $A$ and $B$ by first distilling maximally entangled states between connecting $A$, $C$ and $C$, $B$ and then using a standard repeater protocol. A similar statement is false when distillable key (instead of distillable entanglement) is considered. In particular there are bipartite quantum states $\rho_d\in\D\lb\C^{d}\otimes\C^{2}\otimes \C^{d}\otimes \C^{2}\rb$ (see \cite{bauml2015limitations}) from which private key can be extracted at rate close to $1$, but for which the repeatable key rate fulfills $\mathcal{K}_{A\leftrightarrow C\leftrightarrow B}\lb\rho_d,\rho_d\rb\approx 0$.

Here we introduce the private repeater capacity of a pair of quantum channels. This is a channel-version of the repeatable key rate with one intermediate node. Again the two parties $A$ and $B$ communicate via an intermediate communication node $C$ but now use two quantum channels (from $A$ to $C$ and from $C$ to $B$) and arbitrary classical communication (between all three parties) to establish their secret key. 

Note that this is a more realistic scenario than the state-version of \cite{bauml2015limitations}. It is conceivable that in an actual communication scenario the communicating parties have quantum channels to establish the quantum correlations for the creation of a secret key. But then it would be artificial to restrict possible protocols to those creating a number of copies of a fixed quantum state which are then used to obtain a secret key (see \cite{bauml2015limitations}). Here we consider general protocols allowing for different inputs for the quantum channels at each stage of the protocol possibly depending on measurement outcomes and classical information shared at earlier stages. 

Even in this general framework there are channels with non-repeatable private capacity. In particular we give an example of quantum channels (which are derived from the family of states considered in~\cite{bauml2015limitations}) with private capacity $\mathcal{P}_\leftrightarrow$ close to $1$, but arbitrarily small private repeater capacity. We begin with the definition of the private repeater capacity.

\begin{defn}[Repeater coding schemes assisted by classical communication]\hfill \\
Let $T_1:\M_{d_{A'}}\to\M_{d_{C'}}$ and $T_2:\M_{d_{C''}}\to\M_{d_{B'}}$ denote two quantum channels where $C'$ and $C''$ denote systems controlled by a party $C$. A $(k,m^1, m^2,\epsilon)$-repeater coding scheme for private state generation assisted by classical communication (see Figure \ref{fig:repCodScheme}) is given by a word $w\in\lset 1,2\rset^{m}$ for $m=m^1 + m^2$ with $|\lset i: w_i = 1\rset| = m^1$ (and $|\lset i: w_i = 2\rset| = m^2$), a separable initial state 
\[
\rho^{(1)}\in \begin{cases}\text{Sep}_{A_1 A':C_1:B_1}\lb\C^{d_{A_1}d_{A'}}\otimes \C^{d_{C_1}}\otimes \C^{d_{B_1}}\rb , ~\text{if }w_1 = 1  \\
\text{Sep}_{A_1:C_1C':B_1}\lb\C^{d_{A_1}}\otimes \C^{d_{C_1}d_{C'}}\otimes \C^{d_{B_1}}\rb, ~\text{if }w_1 = 2
\end{cases}
\]
and a set of LOCC-operations (w.r.t. the bipartition into $A,B$ and $C$ systems)
\[
L_i:\M_{d_{A_i}}\otimes \M_{d_{C_i}}\otimes \M_{d_{B_i}}\otimes \M_{d_{D_{w_i}}}\ra \M_{d_{A_{i+1}}}\otimes\M_{d_{C_{i+1}}}\otimes \M_{d_{B_{i+1}}}\otimes \M_{d_{E_{w_i}}} 
\]  
for each $i\in\lset 1,\ldots ,m-1\rset$ and 
\[
L_{m}:\M_{d_{A_m}}\otimes \M_{d_{C_m}}\otimes \M_{d_{B_m}}\otimes\M_{d_{D_{w_m}}} \ra \M_{d_{\tilde{A}}}\otimes \M_{d_{\tilde{B}}}.
\]
Here we set $D_{1} = A'$ (i.e.\ a system at party $A$) and $D_{2} = C''$ (i.e.\ a system at party $C$) and in the same way $E_{1} = C'$ (i.e.\ a system at party $C$) and $E_{2} = B'$ (i.e.\ a system at party $B$). The dimensions $d_{A_i}, d_{B_i}, d_{C_i}\in\N$ and $d_{\tilde{A}} = d_{\tilde{B}}$ are arbitrary. Furthermore, we require the output state
\[
\phi^{(m^1,m^2)}_{\tilde{A}\tilde{B}} = L_{m}\circ\prod^{m-1}_{i=1} \lb T_{w_i}^{D_{w_i}\ra E_{w_i}}\circ L_i\rb\circ T_{w_1}^{D_{w_1}\ra E_{w_1}}\lb\rho^{(1)}\rb
\]
to fulfill 
\[
\epsilon = \frac{1}{2}\| \phi^{(m^1,m^2)}_{\tilde{A}\tilde{B}} - \gamma_{\tilde{A}\tilde{B}}\|_1 
\]
for a private state $\gamma_{\tilde{A}\tilde{B}}$ with $2^k$-dimensional key part (see Definition \ref{defn:PrivState}).

\label{defn:RepCodingScheme}
\end{defn}
   
Note that the order (and number) of channel applications (specified by the word $w$) in the protocols from Definition \ref{defn:RepCodingScheme} is deterministic in the sense, that it cannot depend on outcomes of measurements made during the protocol. This is to avoid the complications from determining the rate of a protocol where the order and number of channel applications is not fixed.

\begin{figure}
\includegraphics[scale=0.6]{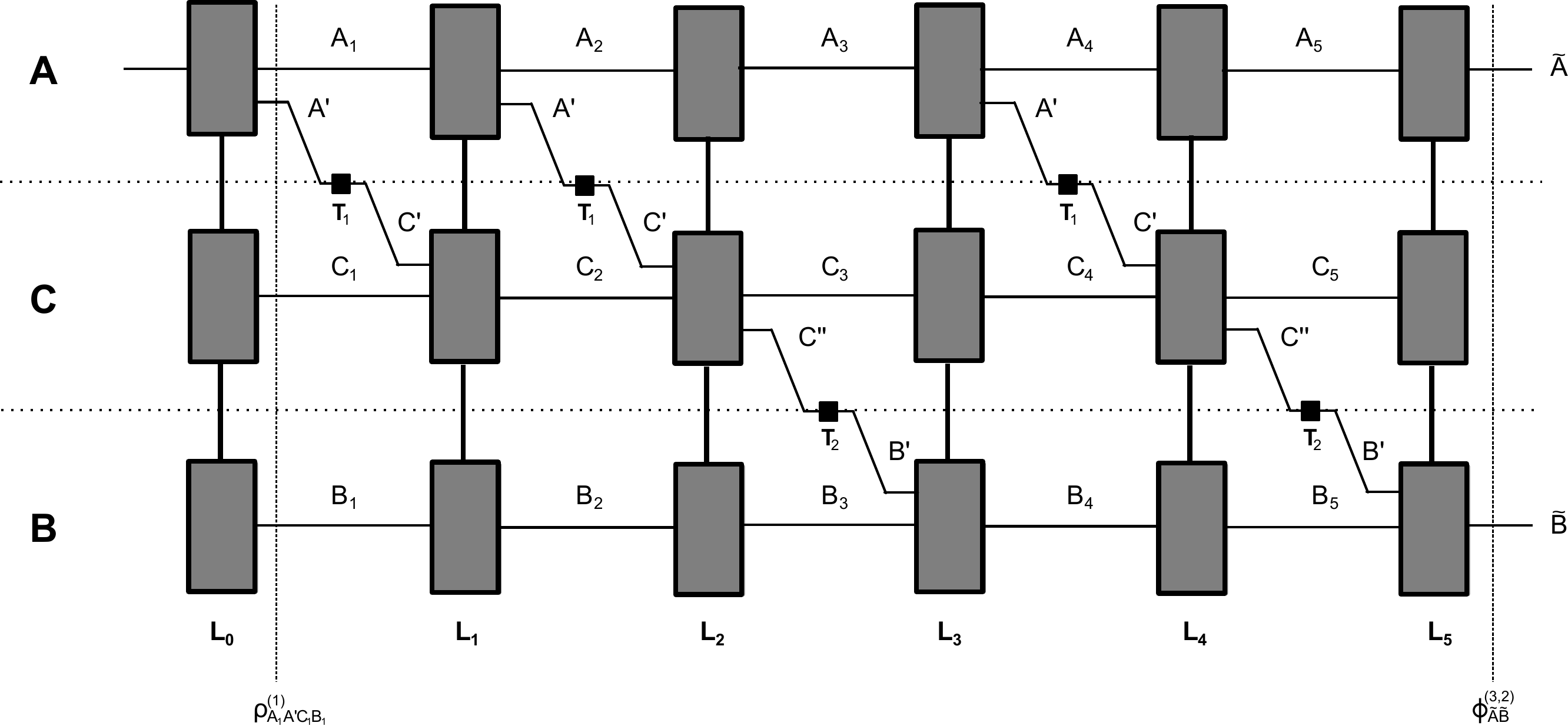}
\caption{Repeater coding scheme for private state generation assisted by classical communication (cf. Definition \ref{defn:RepCodingScheme}) in the case of $m^1=3$ uses of the channel $T_1:\M_{d_{A'}}\ra\M_{d_{C'}}$ and $m^2=2$ uses of the channel $T_2:\M_{d_{C''}}\ra \M_{d_{B'}}$ and channel order $w = (1,1,2,1,2)$. Here $L_0$ denotes an LOCC-operation used to create the separable initial state $\rho^{(1)}_{A_1 A' B_1}$.}
\label{fig:repCodScheme}
\end{figure}

\begin{defn}[Repeated private capacity assisted by classical communication]
We call $R\in\R^+$ an achievable rate for repeated private communication over the quantum channels $T_1:\M_{d_{A'}}\ra\M_{d_{C'}}$ and $T_2:\M_{d_{C''}}\ra\M_{d_{B'}}$ assisted by classical communication iff for each $\nu\in\N$ there exists a $(k_\nu,m^1_\nu,m^2_\nu,\varepsilon_\nu)$-repeater coding scheme for private state generation assisted by classical communication (as in Definition \ref{defn:RepCodingScheme}) with $m^1_\nu, m^2_\nu\ra\infty$ as $\nu\ra\infty$ such that $R=\min\lb\lim_{\nu\to\infty}\frac{k_\nu}{m^1_\nu}, \lim_{\nu\to\infty}\frac{k_\nu}{m^2_\nu}\rb$ and $\lim_{\nu\to\infty}\varepsilon_\nu=0$. The repeated private capacity $\mathcal{P}_{A\leftrightarrow C\leftrightarrow B}(T_1,T_2)$ is defined to be the supremum of all such achievable rates.
\label{defn:PrivCap}
\end{defn}

Before stating our main result we will discuss some properties of the repeated private capacity. For quantum channels $T_1:\M_{d_{A'}}\ra\M_{d_{C'}}$ and $T_2:\M_{d_{C''}}\ra\M_{d_{B'}}$ consider a sequence of coding schemes for $\mathcal{P}_{A\leftrightarrow C\leftrightarrow B}(T_1,T_2)$ achieving a rate $R>0$. By combining the parties $A$ and $C$ (or $C$ and $B$) any such sequence can be transformed into a sequence of coding schemes for $\mathcal{P}_\leftrightarrow\lb T_2\rb$ (or $\mathcal{P}_\leftrightarrow\lb T_1\rb$) achieving at least the same rate $R>0$. Therefore the following bound holds
\begin{equation}
\mathcal{P}_{A\leftrightarrow C\leftrightarrow B}(T_1,T_2) \leq \min\lb \mathcal{P}_\leftrightarrow\lb T_1\rb, \mathcal{P}_\leftrightarrow\lb T_2\rb\rb.
\label{equ:trivBoundRep}
\end{equation}
We also have the following lemma similar to \cite[Lemma 12]{bauml2015limitations}:

\begin{lem}[Transposition trick]

Let $T_1:\M_{d_{A'}}\ra\M_{d_{C'}}$ and $T_2:\M_{d_{C''}}\ra\M_{d_{B'}}$ be two quantum channels such that $\vartheta_{d_{C'}}\circ T_1$ and $T_2\circ \vartheta_{d_{C''}}$ are quantum channels as well (here $\vartheta_d:\M_{d}\ra\M_d$ denotes the matrix transposition in any fixed basis). Then we have
\[
\mathcal{P}_{A\leftrightarrow C\leftrightarrow B}(T_1,T_2) = \mathcal{P}_{A\leftrightarrow C\leftrightarrow B}(\vartheta_{d_{C'}}\circ T_1,T_2\circ \vartheta_{d_{C''}}).
\]
\label{lem:transpTrick}
\end{lem}

\begin{proof}

The proof goes by transforming any protocol for the channels $T_1$ and $T_2$ into a protocol for the channels $\tilde{T}_1 = \vartheta_{d_{C'}}\circ T_1$ and $\tilde{T}_2 = T_2\circ \vartheta_{d_{C''}}$ leaving the output state unchanged. For $m^1,m^2\in\N$ consider a word $w\in\lset 1,2\rset^{m}$ for $m=m^1 + m^2$ with $|\lset i: w_i = 1\rset| = m^1$ (and $|\lset i: w_i = 2\rset| = m^2$). Now consider a protocol for repeated private state generation over the quantum channels $T_1:\M_{d_{A'}}\ra\M_{d_{C'}}$ and $T_2:\M_{d_{C''}}\ra\M_{d_{B'}}$ assisted by classical communication as in Definition \ref{defn:RepCodingScheme} where $w$ specifies the order of channel uses. This protocol is given by a the set of LOCC-operations $\lset L_i\rset^{m}_{i=1}$ (w.r.t. to the parties $A,B$ and $C$) and initial state $\rho^{(1)}$, creating the output state (see Definition \ref{defn:RepCodingScheme})
\[
\phi^{(m^1,m^2)}_{\tilde{A}\tilde{B}} = L_{m}\circ\prod^{m-1}_{i=1} \lb T_{w_i}^{D_{w_i}\ra E_{w_i}}\circ L_i\rb\circ T_{w_1}^{D_{w_1}\ra E_{w_1}}\lb\rho^{(1)}\rb. 
\]
For each $i\in\lset 1,\ldots ,m-1\rset$ we can define new LOCC-operations by 
\[
\tilde{L}_i = \vartheta_{\hat{C}_{i+1}}\circ L_i\circ \vartheta_{\hat{C}'_i}
\] 
where we denote by $\hat{C}'_{i}$ all systems at party $C$ in step $i$ after the channel (either $T_1$ or $T_2$) has been applied (see Definition \ref{defn:RepCodingScheme}). Similarly we denote by $\hat{C}_{i+1}$ all systems at party $C$ before the channel has been applied. The $\tilde{L}_i$ are indeed LOCC-operations, which can be seen from writing $L_i$ in its Kraus-decomposition (according to \eqref{LOCCkraus}) and applying the partial transpositions. In the final step we define 
\[
\tilde{L}_{m} = L_{m}\circ \vartheta_{\hat{C}'_{m}},
\]
which is again LOCC (w.r.t. to the $A,B$ and $C$ systems) as there is no $C$ system at the output of this map. We also define a new initial state $\tilde{\rho}^{(1)}$ by
\begin{equation}
\tilde{\rho}^{(1)} = \vartheta_{\hat{C}_1}\lb\rho^{(1)}\rb, 
\label{equ:initialStateTransf}
\end{equation}
which is a state since $\rho^{(1)}$ was chosen to be separable (see Definition \ref{defn:RepCodingScheme}).

Now note that the LOCC-operations $\lset \tilde{L}_i\rset^{m}_{i=1}$ with initial state $\tilde{\rho}^{(1)}$ define a new protocol for repeated private state generation (with the same word $w$ as before) for the transposed channels $\tilde{T}_1= \vartheta_{d_{C'}}\circ T_1$ and $\tilde{T}_2 = T_2\circ \vartheta_{d_{C''}}$. The output state of the new protocol can be computed and is given by
\begin{align*}
\tilde{\phi}^{(m^1,m^2)}_{\tilde{A}\tilde{B}} &= \tilde{L}_{m}\circ\prod^{m-1}_{i=1} \lb \tilde{T}_{w_i}^{D_{w_i}\ra E_{w_i}}\circ \tilde{L}_i\rb\circ \tilde{T}_{w_1}^{D_{w_1}\ra E_{w_1}}\lb\tilde{\rho}^{(1)}\rb \\ 
&= L_{m}\circ \vartheta_{\hat{C}'_{m}}\circ \prod^{m-1}_{i=1} \lb \tilde{T}_{w_i}^{D_{w_i}\ra E_{w_i}}\circ \vartheta_{\hat{C}_{i+1}}\circ L_i\circ \vartheta_{\hat{C}'_i}\rb\circ \tilde{T}_{w_1}^{D_{w_1}\ra E_{w_1}}\lb\vartheta_{\hat{C}_1}\lb\rho^{(1)}\rb\rb \\
&= L_{m}\circ\prod^{m-1}_{i=1} \lb T_{w_i}^{D_{w_i}\ra E_{w_i}}\circ L_i\rb\circ T_{w_1}^{D_{w_1}\ra E_{w_1}}\lb\rho^{(1)}\rb = \phi^{(m^1,m^2)}_{\tilde{A}\tilde{B}}
\end{align*}
where we used that 
\[
\vartheta_{\hat{C}'_i}\circ \tilde{T}_{w_i}^{D_{w_i}\ra E_{w_i}}\circ\vartheta_{\hat{C}_i} = T^{D_{w_i}\ra E_{w_i}}_{w_i}
\]
for each $i\in\lset 1,\ldots ,m\rset$. This shows that any protocol for the channels $T_1$ and $T_2$ corresponds to a protocol for the channels $\vartheta_{d_{C'}}\circ T_1$ and $T_2\circ \vartheta_{d_{C''}}$ with the same output state and hence the same error. Therefore, the achievable rates for both scenarios are the same and so are their capacities. 

\end{proof}

We will need a particular state constructed in \cite{bauml2015limitations}. Consider the state $\rho_d\in\D\lb\C^{d_{A'}}\otimes \C^{d_A}\otimes \C^{d_{B'}}\otimes \C^{d_B}\rb$ for $d_{A}= d_{B} = d$ and $d_{A'}=d_{B'} = 2$ defined as 
\begin{equation}
\rho_d = \frac{1}{2}\begin{pmatrix}
(1-p(d))\frac{\one_d}{d}\otimes \frac{\one_d}{d} & 0 & 0 & (1-p(d))X \\
0 & p(d)\sqrt{YY^\dagger} & 0 & 0 \\
0 & 0 & p(d)\sqrt{Y^\dagger Y} & 0 \\
(1-p(d))X^\dagger & 0 & 0 & (1-p(d))\frac{\one_d}{d}\otimes \frac{\one_d}{d} 
\end{pmatrix}.
\label{equ:approxPBit}
\end{equation}
Here we used $p(d) = \frac{1}{\sqrt{d}+1}$ and matrices 
\begin{align*}
X &= \frac{1}{d\sqrt{d}}\sum^d_{i,j = 1} u_{ij}\proj{ij}{ji} \\
Y &= \frac{1}{d}\sum^d_{i,j = 1} u_{ij}\proj{ii}{jj}
\end{align*}
where $U = (u_{ij})_{ij}$ denotes the quantum Fourier transform given by 
\[
U\ket{k} = \frac{1}{\sqrt{d}}\sum^d_{j=1} e^{2\pi ijk/d}\ket{j}. 
\]
The state $\rho_d$ has been constructed such that it has positive partial transpose, but it is also close to a private state. More specifically we have
\[
\|\rho_d - \gamma_2\|_1 \leq 2p(d) = \frac{2}{\sqrt{d}+1}
\] 
for the private state
\[
\gamma_2 = \frac{1}{2}\begin{pmatrix}
\frac{\one_d}{d}\otimes \frac{\one_d}{d} & 0 & 0 & X \\
0 & 0 & 0 & 0 \\
0 & 0 & 0 & 0 \\
X^\dagger & 0 & 0 & \frac{\one_d}{d}\otimes \frac{\one_d}{d} 
\end{pmatrix}
\]
with 2-dimensional key part (see \cite{bauml2015limitations}). Now we can state the main result of this section:

\begin{thm}[Non-repeatable private capacity]

There is a quantum channel $T_d:\M_{2}\otimes\M_{d}\ra\M_{2}\otimes \M_{d}$ such that
\[
\mathcal{P}_{\leftrightarrow}\lb T_d\rb \geq 1 - h_2\lb\frac{1}{\sqrt{d} + 1}\rb\ra 1,
\]
but 
\[
\mathcal{P}_{A\leftrightarrow C \leftrightarrow B}\lb T_d,T_d\rb\leq \log_2\lb 1+\frac{1}{\sqrt{d} + 1}\rb \ra 0
\]
as $d\ra\infty$. Here $h_2(x) = -x\log_2(x)-(1-x)\log_2(1-x)$ denotes the binary entropy.

\end{thm}

\begin{proof}

Note that $\text{tr}_{BB'}\lb\rho_{d}\rb = \frac{\one_2}{2}\otimes \frac{\one_d}{d}$, which implies that $\rho_d$ (see \eqref{equ:approxPBit}) is the Choi matrix~\cite{choi1975completely} of a quantum channel $T_d$. Moreover, since $\rho_d$ has positive partial transpose both linear maps $\vartheta_{2d}\circ T_d$ and $T_d\circ\vartheta_{2d}$ are also quantum channels. The private capacity of $T_d$ fulfills
\[
\mathcal{P}_{\leftrightarrow}\lb T_d\rb\geq \mathcal{K}_{\leftrightarrow}\lb\rho_d\rb\geq 1 - h_2(\frac{1}{\sqrt{d} + 1}) 
\]
where the second inequality has been proved in \cite[p. 27]{bauml2015limitations}. In order to show the second statement in the theorem we note that by Lemma \ref{lem:transpTrick} and \eqref{equ:trivBoundRep} 
\[
\mathcal{P}_{A\leftrightarrow C \leftrightarrow B}\lb T_d,T_d\rb = \mathcal{P}_{A\leftrightarrow C \leftrightarrow B}\lb \vartheta_{2d}\circ T_d,T_d\circ\vartheta_{2d}\rb\leq \mathcal{P}_{\leftrightarrow}\lb \vartheta_{2d}\circ T_d\rb.
\]  
By the non-regularized bound from Corollary \ref{cor:RegUpperBoundP2} and the simpler bound from Theorem \ref{thm:SimplerUpperBound} we have
\[
\mathcal{P}_{\leftrightarrow}\lb \vartheta_{2d}\circ T_d\rb\leq E_{\max}\lb \vartheta_{2d}\circ T_d\rb \leq B_{\max}\lb \vartheta_{2d}\circ T_d\rb\leq D_{\max}\lb \rho^{T_{B'B}}_{d}\| C_S\rb
\] 
where we choose the separable Choi matrix 
\[
C_S = \frac{1}{2(1+p(d))}\begin{pmatrix}
(1-p(d))\frac{\one_d}{d}\otimes \frac{\one_d}{d} & 0 & 0 & 0 \\
0 & 2p(d)\sqrt{YY^\dagger} & 0 & 0 \\
0 & 0 & 2p(d)\sqrt{Y^\dagger Y} & 0 \\
0 & 0 & 0 & (1-p(d))\frac{\one_d}{d}\otimes \frac{\one_d}{d} 
\end{pmatrix}.
\]
It can be easily checked that $C_S$ is the Choi matrix of an entanglement-breaking channel $S$. Note that
\[
\rho^{T_{B'B}}_{d} = \frac{1}{2}\begin{pmatrix}
(1-p(d))\frac{\one_d}{d}\otimes \frac{\one_d}{d} & 0 & 0 & 0 \\
0 & p(d)\sqrt{YY^\dagger} & p(d)Y & 0 \\
0 & p(d)Y^\dagger & p(d)\sqrt{Y^\dagger Y} & 0 \\
0 & 0 & 0 & (1-p(d))\frac{\one_d}{d}\otimes \frac{\one_d}{d}. 
\end{pmatrix}
\] 
and a straightforward computation shows that $D_{\max}\lb \rho^{T_{B'B}}_{d}\| C_S\rb\leq \log_2(1+p(d))$. This implies that 
\[
\mathcal{P}_{A\leftrightarrow C \leftrightarrow B}\lb T_d,T_d\rb\leq \log_2(1+p(d)) = \log_2\lb 1+\frac{1}{\sqrt{d}+1}\rb.
\] 

\end{proof}

\section{Conclusion}

We established a new inequality involving the sandwiched $\alpha$-R\'{e}nyi divergences and used it to study private communication via quantum channels assisted by classical communication. Specifically, we proved a strong-converse bound on the private capacity assisted by unlimited classical two-way communication. Moreover, this is the first such bound that is non-lockable. We exploited this fact to provide examples of quantum channels for which our bound improves on the transposition bound (Theorem \ref{thm:transpBound}), the squashed entanglement bound (Theorem \ref{thm:squashBound}) and the entanglement cost bound (Theorem \ref{thm:ECostBound}). Furthermore, we used the bound to analyze a quantum repeater version of the private capacity. 

There are some open problems and directions of future research. The main open problem is to show that the relative entropy of entanglement of a quantum channel (instead of the max-relative entropy of entanglement, see \eqref{equ:maxRelEntrQChan}) is an upper bound (and possibly a strong-converse bound) on $P_\leftrightarrow$. So far, this bound has only been shown for teleportation-covariant quantum channels~\cite{pirandola2015ultimate}. Such a result might be obtained from the bound in Theorem \ref{thm:SCBound} (or Corollary \ref{cor:RegUpperBoundP2}) using a smoothing technique (cf. \cite{brandao2010generalization}).    

It should be noted that quantities similar to \eqref{equ:maxRelEntrQChan} for different entanglement measures (replacing the max-relative entropy of entanglement) based on the sandwiched $\alpha$-R\'{e}nyi divergences have been studied before. In \cite{7282883} the $\alpha$-Rains information of a quantum channel (based on a generalization of the Rains bound on distillable entanglement~\cite{rains2001semidefinite}) has been introduced. Here instead of optimizing over separable states leading to an relative entropy of entanglement (cf. Definition \ref{defn:alphaRelEntrEnt}) the optimization runs over a larger set (the so called Rains set) of positive matrices (see \cite{7282883} for details). To our knowledge it is not known whether the $\alpha$-Rains information (for any $\alpha\geq 1$) gives a strong converse bound (or even an upper bound) on $\mathcal{Q}_\leftrightarrow$. For $\alpha = \infty$ this follows almost from our work. The only problem seems to be in the final part of the proof of Theorem \ref{thm:Shifting}, where we cannot reduce the quantity involving the three systems $A,B'$ and $B$ to the Rains information (only involving two systems).

Finally, we should say that the main results from this paper can be extended to infinite dimensional systems using the general framework of non-commutative $L_p$-spaces~\cite{pisier2003non}. This will be contained in future work. 


\section{Acknowledgements}

We thank Mario Berta, Roberto Ferrara, Jedrzej Kaniewski, Christian Majenz, Milan Mosonyi and David Reeb for useful comments and interesting discussions.
We acknowledge financial support from the European Research Council (ERC Grant Agreement
no 337603), the Danish Council for Independent Research (Sapere Aude), the Swiss National
Science Foundation (project no PP00P2 150734) and VILLUM FONDEN via the QMATH Centre of Excellence (Grant No. 10059).

\appendix

\section{Implementability of quantum channels via LOCC operations}

Here we study the class of quantum channels implementable via LOCC-operations from a bipartite state shared between the communicating parties. For such channels the interactive protocols of Definition \ref{defn:CodingScheme} reduce to protocols involving only LOCC-operations performed on copies of the fixed state used for the implementation (see \cite{bennett1996mixed,pirandola2015ultimate}). It is easy to see \cite{pirandola2015ultimate,wilde2016converse} that the distillable entanglement (key) of this state gives an upper bound on the performance of such protocols in the cases of quantum (private) communication.

The reduction of protocols described above is especially interesting when the state used for implementation of the quantum channel is itself preparable using the quantum channel exactly once (see below for a precise definition). This holds e.g.~for teleportation-covariant channels (see \cite{pirandola2015ultimate}). In this case the capacities $\mathcal{Q}_\leftrightarrow$ and $\mathcal{P}_\leftrightarrow$ of the channel are equal to (not only upper bounded by) the distillable entanglement and distillable key respectively of the state used for implementation. Moreover, since this state can be produced using the channel, entanglement measures (e.g.~squashed entanglement, relative entropy of entanglement, etc.) of the state can be related to the corresponding quantities of the channel (see also Theorem \ref{thm:TeleportEntMeas} below). In this way \cite{pirandola2015ultimate, wilde2016converse} derive their upper bounds on the private capacity for particular classes of channels. 

The quantum channels implementable from states using the teleportation protocol have been characterized in~\cite{bowen2001teleportation}. However, in the case of general protocols such a characterization is still missing, and it is not known which quantum channels can be implemented in this way. Here we give an example of a quantum channel, which cannot be implemented by any LOCC-protocol using a state preparable by only a single use of the quantum channel itself. We begin with a definition:  

\begin{defn}
We call a quantum channel $T:\M_{d_A}\ra\M_{d_B}$ \textbf{implementable from its image} if there exists a bipartite quantum state $\sigma_{A''A'}\in\D\lb\C^{d_{A''}}\otimes \C^{d_{A'}}\rb$ for some $d_{A''}\in\N$ and $d_{A'} = d_{A}$ and an LOCC-operation $\Lambda: \M_{d_{A}d_{A''}}\otimes \M_{d_{B'}}\ra \M_{d_B}$ for $d_{B'}=d_{B}$ with respect to the bipartition into $A$ and $B$ systems such that 
\begin{equation}
T^{A\ra B}\lb\rho_{A}\rb = \Lambda^{AA'':B'\ra B}\lb\rho_{A} \otimes (\ident_{A''}\otimes T^{A'\ra B'})\lb\sigma_{A''A'}\rb \rb
\label{equ:Imple}
\end{equation}
for any $\rho_A \in\D_{d_A}$.
\label{defn:ImpleFrIm}
\end{defn} 

Consider an LOCC-monotone $E^{A:B}:\D\lb\C^{d_A}\otimes \C^{d_B}\rb\ra\R_0^+$ for bipartite states. Formally, $E^{A:B}$ is a family of functions depending on the dimensions $d_A$ and $d_B$ decreasing under LOCC-operations applied to the input (LOCC with respect to the chosen bipartition $A:B$). To simplify notation we will omit the dependence on the dimensions. Now we define an associated quantity for quantum channels $T:\M_{d_A}\ra\M_{d_B}$ by setting
\[
E\lb T\rb = \sup_{\rho_{A'A}} E^{A':B}\lb\lb\ident_{A'}\otimes T^{A\ra B}\rb\lb\rho_{A'A}\rb\rb.
\]  
where the supremum is over states $\rho_{A'A}\in\mathcal{D}\lb\C^{d_{A'}}\otimes \C^{d_A}\rb$ with arbitrary $d_{A'}\in\N$ (note that this quantity is not finite in general, but it will be in the examples we consider). We have the following simple consequence for quantum channels implementable from their image:

\begin{thm}
For any LOCC-monotone $E^{A:B}$ and any quantum channel $T:\M_{d_A}\ra\M_{d_B}$ implementable from its image, i.e. of the form \eqref{equ:Imple} for some state $\sigma_{A''A'}\in\D\lb \C^{d_{A''}}\otimes \C^{d_{A'}}\rb$, we have
\[
E\lb T\rb = E^{A'':B'}\lb \lb\ident_{A''}\otimes T^{A'\ra B'}\rb\lb\sigma_{A''A'}\rb\rb.
\]
\label{thm:TeleportEntMeas}
\end{thm}

\begin{proof}
The inequality ``$\geq$'' is clear. As $E^{A:B}$ is an LOCC-monotone we have
\begin{align*}
E\lb T\rb &= \sup_{\rho_{A'''A}} E^{A''':B}\lb\lb\ident_{A'''}\otimes T^{A\ra B}\rb\lb\rho_{A'''A}\rb\rb \\
&= \sup_{\rho_{A'''A}} E^{A''':B}\lb \lb\ident_{A'''}\otimes \Lambda^{AA'':B'\ra B}\rb\lb\rho_{A'''A} \otimes (\ident_{A''}\otimes T^{A'\ra B'})\lb\sigma_{A''A'}\rb \rb\rb \\
&\leq \sup_{\rho_{A'''A}} E^{A'''AA'':B'}\lb \rho_{A'''A} \otimes (\ident_{A''}\otimes T^{A'\ra B'})\lb\sigma_{A''A'}\rb \rb \\
& = E^{A'':B'}\lb(\ident_{A''}\otimes T^{A'\ra B'})\lb\sigma_{A''A'}\rb \rb.
\end{align*}
Here the last equality follows from the fact that removing or adding a local uncorrelated system is an LOCC-operation. 
\end{proof}

In the following we will only evaluate the LOCC-monotones $E_R$ and $E_{sq}$ on bipartite states where the systems in the bipartition are clear from context. Therefore, we will omit the indices denoting these systems to simplify notation. Now we can present the main result of this appendix:

\begin{thm}
There exists a quantum channel $T:\M_{d_A}\ra\M_{d_B}$ for some dimensions $d_A,d_B\in\N$ that is not implementable from its image, i.e.~there is no state $\sigma_{A''A'}$ and LOCC protocol $\Lambda$ such that $T$ can be written as in \eqref{equ:Imple}.
\label{thm:nonImpl}
\end{thm}

For the proof we will need some special states. The antisymmetric state $\alpha_d\in\D\lb\C^d\otimes \C^d\rb$ for $d\geq 2$ is defined as 
\[
\alpha_d = \frac{1}{d(d-1)}\lb\one_d\otimes \one_d - \mathbb{F}_d\rb.
\]
In \cite[Lemma 6]{christandl2012entanglement} it is shown that for even $d\in\N$
\begin{equation}
E_{sq}\lb\alpha_d\rb\leq \log_2\lb \frac{d+2}{d}\rb.
\label{equ:squashedEntAnti}
\end{equation}
It has also been shown in \cite[Corollary 3]{christandl2012entanglement} that for every $d\geq 2$  we have
\[
\lim_{n\ra\infty} \frac{1}{n}E_R\lb\alpha^{\otimes n}_d\rb \geq \log_2\lb\sqrt{\frac{4}{3}}\rb.
\]
Clearly, for any $\epsilon >0$ this implies the existence of an $N_\epsilon\in\N$ such that 
\begin{equation}
E_R\lb\alpha^{\otimes n}_d\rb \geq n \lb\log_2\lb \sqrt{\frac{4}{3}}\rb -\epsilon\rb
\label{equ:relativeEntAnti}
\end{equation}
for all $n\geq N$. We will also use the flower states $\rho^f_d$ from Section \ref{sec:Flower} considered as bipartite states with respect to the bipartition into A and B systems (both $2d$ dimensional, see \eqref{equ:flowerState}). Note that the squashed entanglement of the flower states has an easy formula (see \eqref{equ:FlowerSquash}). Furthermore, as the partial trace $\text{tr}_{B'}\lb\rho^f_d\rb$ over the $2$-dimensional $B'$ system is separable we have (using non-lockability of $E_R$, see \cite{horodecki2005locking}) that
\begin{equation}
E_R\lb\rho^f_d\rb \leq 2 .
\label{equ:relativeFlower}
\end{equation}

Finally observe that for any $n,l\in\N$ and dimension $d=2^n l^n$ we have
\begin{equation}
\tau_0 := \alpha^{\otimes n}_{2l}\in\D\lb\C^d\otimes \C^d\rb
\label{equ:tau0}
\end{equation}
and 
\begin{equation}
\tau_1 := \rho^f_{2^{n-1}l^n}\in\D\lb\C^d\otimes \C^d\rb.
\label{equ:tau1}
\end{equation}
Using the formulas for $E_R$ and $E_{sq}$ from above, and additivity of the squashed entanglement (see \cite[Proposition 4]{christandl2004squashed}) we compute (with $N_{\frac{1}{2}}$ defined before \eqref{equ:relativeEntAnti})
\begin{align*}
E_R\lb \tau_0\rb &\geq n\lb\log_2\lb\sqrt{\frac{4}{3}}\rb - \frac{1}{2}\rb \geq \frac{1}{2}n \text{ for all }n\geq N_{1/2} \\
E_R\lb\tau_1\rb & \leq 2 \\
E_{sq}\lb \tau_0\rb &= n E_{sq}\lb\alpha_{2l}\rb \leq n\log_2\lb 1 + \frac{1}{l}\rb \\
E_{sq}\lb \tau_1\rb &= \frac{1}{2}+\frac{1}{2}n + \frac{1}{2}n\log_2\lb l\rb.
\end{align*}
Therefore, choosing $n,l$ large enough we have
\begin{align}
E_R\lb\tau_1\rb &\ll E_R\lb\tau_0\rb 
\label{equ:dichotomy1}\\
E_{sq}\lb\tau_1\rb &\gg E_{sq}\lb\tau_0\rb.
\label{equ:dichotomy2}
\end{align}

\begin{proof}[Proof of Theorem \ref{thm:nonImpl}]
Choose $n,l\in\N$ large enough such that the states $\tau_0, \tau_1\in\D\lb\C^d\otimes \C^d\rb$ (see \eqref{equ:tau0} and \eqref{equ:tau1}) for $d=2^n l^n$ satisfy \eqref{equ:dichotomy1} and \eqref{equ:dichotomy2}. Now define two channels $T_0,T_1:\M_d\ra\M_d$ with Choi matrices $C_{T_0}=\tau_0$ and $C_{T_1} = \tau_1$ (note that these maps are indeed trace-preserving). 

The quantum channel $T_0$ is teleportation implementable (as the channel corresponding to the antisymmetric state is Weyl-covariant), i.e. it is of the form \eqref{equ:Imple} with $\sigma = \tau_0$ and $\Lambda$ the teleportation protocol (see \cite{bennett1996mixed}). Therefore, we can apply Theorem \ref{thm:TeleportEntMeas} to conclude that 
\begin{align*}
E_R\lb T_0\rb &= E_R\lb\tau_0\rb \\
E_{sq}\lb T_0\rb &= E_{sq}\lb \tau_0\rb
\end{align*}
Let $B'$ denote the 2-dimensional part of the output system of $T_1$ corresponding to the $B'$ system of $\tau_1$ (which is a flower state, see discussion following \eqref{equ:flowerState}). Then $\text{tr}_{B'}\circ T_1$ is entanglement-breaking (as its Choi matrix is separable) and using non-lockability of $E_R$ (see \cite{horodecki2005locking}) and the equations above we have
\begin{equation}
E_R\lb T_1\rb \leq 2 \ll E_R\lb\tau_0\rb = E_R\lb T_0\rb.
\label{equ:relEntrT0T1}
\end{equation} 
For the squashed entanglement we obtain
\begin{equation}
E_{sq}\lb T_0\rb = E_{sq}\lb \tau_0\rb\ll E_{sq}\lb\tau_1\rb\leq E_{sq}\lb T_1\rb.
\label{equ:squashedT0T1}
\end{equation}
Now consider the switch channel $T:\M_d\otimes \M_2\ra\M_{d}\otimes \M_2$ given by 
\[
T = T_0\otimes P_0 + T_1\otimes P_1
\]
with projectors $P_i:\M_2\ra\M_2$ given by $P_i(\rho) = \bra{i}\rho\ket{i}\proj{i}{i}$ for $i\in\lset 0,1\rset$. In the following we denote by ``$a$'' a 2-dimensional system at party $A$ and by ``$b$'' a 2-dimensional system at party $B$. These will denote the switch systems used for the quantum channel $T$. As 
\[
\lb\ident_{A'}\otimes T_i^{A\ra B}\rb\lb\rho_{A'A}\rb\otimes \proj{i}{i}_b = \lb\ident_{A'}\otimes T^{Aa\ra Bb}\rb\lb\rho_{A'A}\otimes \proj{i}{i}_{a}\rb ,
\]
for any $\rho_{A'A}\in\D\lb\C^{d_{A'}}\otimes \C^{d_A}\rb$ and $i\in\lset 0,1\rset$ we conclude that 
\begin{align}
E_R\lb T_i\rb\leq E_R\lb T\rb \label{equ:relEntOrder} \\
E_{sq}\lb T_i\rb \leq E_{sq}\lb T\rb .
\label{equ:squashOrder}
\end{align}

Assume now that $T$ is implementable from its image and let $\sigma_{A''A'a'}\in\D\lb \C^{d_{A''}}\otimes \C^{d_{A'}}\otimes \C^2\rb$ denote the state used for the implementation as in \eqref{equ:Imple}. Note that the dimension $d_{A''}$ is arbitrary and we consider the joint system $A'a'$ as the input (therefore taking the role of $A'$ in \eqref{equ:Imple}) for the channel. We can write 
\[
\sigma_{A''A'a'} = \sum^1_{i=0}\sum^1_{j=0} X_{A''A'}^{ij}\otimes \proj{i}{j}_{a'}
\]
with matrices $X_{A''A'}^{ij}\in \M_{d_{A''}}\otimes \M_{d_{A'}}$. Positivity of $\sigma_{A''A'a'}$ implies positivity of $X_{A''A'}^{00}$ and $X_{A''A'}^{11}$. Now we have
\[
\lb\ident_{A''}\otimes T^{A'a'\ra B'b'}\rb\lb\sigma_{A''A'a'}\rb = \sum^1_{i=0} \lb\ident_{A''}\otimes T_i^{A'\ra B'}\rb\lb X_{A''A'}^{ii}\rb\otimes \proj{i}{i}_{b'} .
\] 
As $\sigma_{A''A'a'}$ is normalized we can write
\begin{align}
&\lb\ident_{A''}\otimes T^{A'a'\ra B'b'}\rb\lb\sigma_{A''A'a'}\rb \nonumber \\
&= p\lb\ident_{A''}\otimes T_0^{A'\ra B'}\rb\lb \sigma_{A''A'}^{0}\rb\otimes \proj{0}{0}_{b'} + (1-p)\lb\ident_{A''}\otimes T_1^{A'\ra B'}\rb\lb \sigma_{A''A'}^{1}\rb\otimes \proj{1}{1}_{b'}
\label{equ:decompTeleState}
\end{align}
for $p = \text{tr}\lb X_{A''A'}^{00}\rb\in\lbr 0,1\rbr$ and states 
\[
\sigma_{A''A'}^i =\begin{cases} \frac{1}{\text{tr}\lb X_{A''A'}^{ii}\rb}X_{A''A'}^{ii}, &\text{ if }\text{tr}\lb X_{A''A'}^{ii}\rb\neq 0\\
0, &\text{ else.} \end{cases}
\]
Note that $p$ only depends on $\sigma$. Now applying Theorem \ref{thm:TeleportEntMeas} (as $E_R$ and $E_{sq}$ are LOCC-monotones) together with \eqref{equ:decompTeleState} and convexity of $E_R$ and $E_{sq}$ (see \cite[Proposition 3]{christandl2004squashed} for the latter) we obtain  
\begin{align}
E_R\lb T\rb &\leq p E_R\lb T_0\rb + (1-p)E_{R}\lb T_1\rb \label{equ:FinalRelEntr} \\
E_{sq}\lb T\rb &\leq p E_{sq}\lb T_0\rb + (1-p)E_{sq}\lb T_1\rb \label{equ:FinalSquash}
\end{align} 
Finally, it follows from \eqref{equ:relEntOrder}, \eqref{equ:FinalRelEntr} and \eqref{equ:relEntrT0T1} that 
\[
E_R\lb T_0\rb\leq E_R\lb T\rb \leq p E_R\lb T_0\rb + (1-p)E_R\lb T_1\rb\leq E_R\lb T_0\rb.
\]
As $E_R\lb T_1\rb\ll E_R\lb T_0\rb$ this implies that $p=1$. The same line of reasoning for the squashed entanglement using \eqref{equ:squashOrder}, \eqref{equ:FinalSquash} and \eqref{equ:squashedT0T1} gives
\[
E_{sq}\lb T_1\rb\leq E_{sq}\lb T\rb \leq p E_{sq}\lb T_0\rb + (1-p)E_{sq}\lb T_1\rb\leq E_{sq}\lb T_1\rb.
\]
As $E_{sq}\lb T_0\rb\ll E_{sq}\lb T_1\rb$ this implies that $p=0$ which is a contradiction to the previous derivation.
\end{proof}

Note that the quantum channel $T$ constructed in the previous example might be implementable using LOCC-operations and a state that can be prepared from two or more uses of the channel. This would be the case for example if the channel $T_1$ (coming from the flower state) would be implementable from its image. The reduction technique for interactive protocols (see \cite{pirandola2015ultimate}) would still apply then, however relating $\mathcal{Q}_\leftrightarrow$ (or $\mathcal{P}_\leftrightarrow$) to distillable entanglement (or distillable key) of a more complicated state. It is then not clear how to obtain e.g.~the bound based on the relative entropy of entanglement of the quantum channel from the methods of \cite{pirandola2015ultimate} without an additional factor depending on the number of channel uses to prepare this state.

\bibliographystyle{IEEEtran}

\end{document}